\documentclass[12pt,letterpaper]{article}
\usepackage{enumitem}

\usepackage{renstyles}

\usepackage{cite}
\hypersetup{
    citecolor=green!50!black,
    linkcolor=blue!80!black
}

\usepackage{pgf}
\usepackage{tikz}
\usetikzlibrary{positioning}
\usepackage{subcaption}
\usepackage{multirow}
\usepackage{array}
\usepackage{algpseudocode}
\usepackage{yhmath}

\SetSymbolFont{stmry}{bold}{U}{stmry}{m}{n}

\makeatletter
\DeclareRobustCommand\widecheck[1]{{\mathpalette\@widecheck{#1}}}
\def\@widecheck#1#2{%
    \setbox\z@\hbox{\m@th$#1#2$}%
    \setbox\tw@\hbox{\m@th$#1%
       \widehat{%
          \vrule\@width\z@\@height\ht\z@
          \vrule\@height\z@\@width\wd\z@}$}%
    \dp\tw@-\ht\z@
    \@tempdima\ht\z@ \advance\@tempdima2\ht\tw@ \divide\@tempdima\thr@@
    \setbox\tw@\hbox{%
       \raise\@tempdima\hbox{\scalebox{1}[-1]{\lower\@tempdima\box
\tw@}}}%
    {\ooalign{\box\tw@ \cr \box\z@}}}
\makeatother

\title{Chosen-Plaintext Attacks of Double Random Phase Encryption with Nonlinear Optical Media}
\author{
    Yan Cheng\thanks{Department of Applied Physics and Applied Mathematics, Columbia University, New York, NY 10027; yc3855@columbia.edu}
    \and
    Yiwei Chen\thanks{Department of Applied Physics and Applied Mathematics, Columbia University, New York, NY 10027; yc4376@columbia.edu}
    \and
    Kui Ren \thanks{Department of Applied Physics and Applied Mathematics, Columbia University, New York, NY 10027; kr2002@columbia.edu}
    \and
    Nathan Soedjak \thanks{Department of Applied Physics and Applied Mathematics, Columbia University, New York, NY 10027; ns3572@columbia.edu}
}
\begin{document}

\maketitle

\begin{abstract}

This paper studies an inverse problem in nonlinear optical encryption. We examine chosen-plaintext attacks (CPA) on a nonlinear optical encryption strategy that integrates double random phase encryption (DRPE) into a nonlinear optical propagation model to enhance the security of the combined system. We first demonstrate that the system's phase information can be decoded from carefully designed differential CPA data. We then demonstrate that the strength of the optical device's nonlinearity can also be recovered from CPA data, indicating that including this parameter as an additional security key does not enhance protection against CPA attacks, although numerical simulations show that strong nonlinearity still poses significant challenges for CPA attacks. Finally, we provide a stability analysis to demonstrate that small errors in decoded security keys result in only small errors in the decrypted text, even though the encryption process is nonlinear.

\end{abstract}


\section{Introduction}

The exponential growth of data traffic and the emerging demand for secure high-speed communication networks have catalyzed interest in optical encryption systems. Unlike traditional digital cryptography, optical encryption offers potential for ultrafast, high-bandwidth, and parallel processing capabilities by leveraging the physical dynamics of light propagation; see, for instance,~\cite{Abuturab-OLE24,AhCa-OLE23,HuChGoZh-OLE20,InCh-JVCIR21,PeZhWeYu-OL06,SiZh-OL04,UnJoSi-IL00,WaNiWaZhHu-OLE22,YuMaXiSuLiJiWa-JMO22,Yu-NC24,ZhLiZhXuXuWa-OE23,ZhTaLi-OE21,ZhXiCh-OL20,ZhXiCh-OLE21,ZhZhWaXuXuZh-OLE23,ZhYaLiLvZhChKeQiSh-OC23,KaSiKa-MPE21,LiGuSh-OLT} and references therein for some samples of recent development in the field. 

One of the promising approaches proposed in recent years is the use of nonlinear optical wave propagation, especially governed by the nonlinear Schr\"{o}dinger equation (NLSE), to achieve encryption through spatiotemporal optical modulation~\cite{AlJaHuMu-OQE23,HuChGoZh-OLE20,HoSi-EL22,MoChKaJaCh-NP23,TuPrLeWaFrKaDe-Optica17}. In this framework, the signal to be encrypted is encoded as an initial condition or a modulation pattern of a nonlinear dispersive wave, typically within an optical fiber or a waveguide array, whose evolution under the NLSE generates a complex output field. This output, in effect, represents the ciphertext, while the encryption key corresponds to parameters in the NLSE system such as dispersion, nonlinearity, or engineered perturbations. The system's inherent nonlinearity and sensitivity to initial conditions suggest potential cryptographic strength, particularly in resisting brute-force and statistical attacks. 

To further enhance the security of the optical system, ~\cite{HuChGoZh-OLE20,HoSi-EL22,MoChKaJaCh-NP23} couple the nonlinearity with the classical double-random phase encoding (DRPE)~\cite{ReJa-OL95} method. The new system uses two phase-only spatial light modulators (SLMs) with independent random phase masks $\varphi_1(\bx)$ and $\varphi_2(\bx)$. A plaintext image is first modulated by the mask $e^{i\varphi_1(\bx)}$ and then propagated through the nonlinear Schr\"odinger equation (NLSE). At an intermediate plane, a second phase mask $e^{i\varphi_2(\bx)}$ is applied, and the resulting field is again propagated through the same nonlinear medium to produce the ciphertext.

While extensive studies have been conducted to analyze the security features of the aforementioned nonlinear optical encryption system, a rigorous understanding of the vulnerability of these systems under various attack models remains limited. 

In this work, we study the aforementioned nonlinear optical encryption scheme under a cryptanalytic threat called the chosen plaintext attacks (CPA) \cite{PeWeZh-OL06}. In a CPA, the attacker is able to inject specific inputs (plaintexts) into the encryption system and observe the corresponding ciphertexts, enabling potential inference of the encryption mechanism. In the optical context, this translates to controlling the input light waveform and observing the field that is propagated at the output. Since the nonlinear Schr\"{o}dinger dynamics are deterministic and structurally rich, potentially invertible under idealized conditions, there is an inherent risk that an attacker may exploit CPA to infer system parameters, inverse scattering features, or reconstruct effective transfer operators.

The rest of this paper is organized as follows. In~\Cref{SEC:DRPE}, we review the mathematical formulation of the nonlinear optical encryption process with DRPE. We then introduce the CPA attack as an inverse problem in~\Cref{SEC:CPA} and analyze a differential CPA strategy to handle the nonlinearity. In~\Cref{SEC:Phase-Medium}, we show that using nonlinearity as an additional security key does not increase the difficulty of CPA attacks. We then present a stability theory in~\Cref{SEC:Stability} to demonstrate
that small errors in decoded security keys result in only small errors in the decrypted text. Finally, numerical simulations are provided in~\Cref{SEC:Numer} to verify the theoretical findings.

\section{DRPE using nonlinear optical media}
\label{SEC:DRPE}

We briefly review the nonlinear Double Random-Phase Encryption (DRPE) scheme proposed by Hou and Situ~\cite{HoSi-EL22}, which forms the starting point of our analysis. Let $f(\bx)$ be the plaintext image we intend to encrypt.
We first modulate the plaintext image by the mask $e^{i\varphi_1(\bx)}$ and then propagate it through a photorefractive medium whose paraxial dynamics are modeled by a nonlinear Schr\"odinger equation (NLSE). A second phase mask $e^{i\varphi_2(\bx)}$ is applied to the output to produce the ciphertext.

Mathematically, the encryption process can be modelled using the following saturating nonlinear Schr\"{o}dinger equation~\cite{deMoPe-CVPDE13,GaHe-JOSA91,HoSi-EL22,ZeBeMiLuCaLiZhZe-OE25}:
\begin{equation}\label{EQ:CNLS}
    \begin{array}{rcll}
	i\dfrac{\partial u}{\partial z} +\dfrac{1}{2k(\bx)}\Delta u + \beta(\bx) \dfrac{|u|^2}{1+|u|^2} u &=& 0, & \mbox{in}\ \ (0, L_z]\times \Omega\\[2ex]
	u(0, \bx)&=&f(\bx)e^{i\varphi_1(\bx)},& \mbox{in}\ \ \Omega
	\end{array}
\end{equation}
where $k(\bx)=2\pi n(\bx)/\lambda$ with $n(\bx)$ being the base refractive index of the medium and $\lambda$ the free-space wavelength of the input image. The propagation is along the positive $z$ direction, $\Delta$ is the standard Laplace operator on the variable $\bx\in\Omega\subset\bbR^d$ in the transverse plane, and $\beta(\bx)$ is a coefficient related to the nonlinear properties of the medium. For the most part of the work, we take $\Omega = [0,L_x] \times [0,L_y] \subset \bbR^2$ and impose the periodic boundary conditions: 
\begin{equation}\label{EQ:NLS BC}
   \begin{array}{rcll}
	    u(z, 0,y) &=& u(z, L_x,y),& \mbox{in}\ \ [0, L_z]\times [0, L_y]\\
	    u(z, x,0) &=& u(z, x, L_y),& \mbox{in}\ \ [0, L_z]\times [0, L_x]
   \end{array}
\end{equation}
The process leads to the signal observed at $L_z$ as
\begin{equation}
    g(\bx):=e^{i\varphi_2(\bx)} u(L_z, \bx)=e^{i\varphi_2(\bx)} \Lambda_{k, \beta}^{L_z} \Big[f(\bx)e^{i\varphi_1(\bx)}\Big]\,,
\end{equation}
where we use
\begin{equation}
    \Lambda_{k, \beta}^{L_z}: u(0, \bx)\mapsto u(L_z, \bx)
\end{equation}
to denote the initial-to-terminal condition map of equation~\eqref{EQ:CNLS}. The saturable nonlinearity $|u|^2/(1+|u|^2)$ arises from the refractive index change in the SBN:61 photorefractive crystal used in the optical setup of~\cite{HoSi-EL22}. The authors of \cite{HoSi-EL22} argue that this nonlinearity can ``scramble’’ the spatial frequency content of the field more effectively than linear DRPE, thereby enhancing security.

\begin{remark}\label{REM:simplification}
The original scheme of \cite{HoSi-EL22} applies a second propagation after the mask $\varphi_2$. Since this propagation uses no additional secret parameters, an attacker with access to the complex ciphertext can invert it directly. We therefore analyze the simplified model without loss of generality; the security of the system rests entirely on the phase masks.
\end{remark}

The decryption process aims to recover $f$ from the encrypted message $g$. This can be done easily by backpropagating $g$ from $z=L_z$ to $z=0$. One first removes the second phase mask and then evolves the NLSE backward in $z$, followed by removal of the first mask. Formally, decryption requires solving
\begin{equation}\label{EQ:decrypt-clean}
    i\partial_z u + \frac{1}{2k}\Delta u
        + \beta\,\dfrac{|u|^2}{1+|u|^2}\,u = 0, \qquad u({L_z},\bx) = g(\bx)e^{-i\varphi_2(\bx)},
\end{equation}
backwards from $z=L_z$ to $z=0$ to recover
\[
    f(\bx) =u(0,\bx)\,e^{-i\varphi_1(\bx)}.
\]
In practice, because the equation is unchanged under the transform $u \mapsto \overline u$ and $z\mapsto -z$, after taking the complex conjugate and reversing time, we have the same forward evolution problem.

This procedure assumes that the NLSE evolution is reversible and that numerical or experimental errors do not destroy information during the nonlinear wave propagation, though that may be challenging in practice \cite{SaDiGoFi-PhysicaD20}. Hou and Situ showed via experiments that, despite the presence of nonlinearity, the underlying photorefractive dynamics are sufficiently stable to allow practical decryption \cite{HoSi-EL22}. Therefore, the main question to be addressed is whether this encryption process is secure against attacks. Such attacks are essentially inverse problems, where the goal is to recover the security keys $\varphi_1$ and $\varphi_2$, assuming the attacker has access to the terminal data (i.e., the ciphertext).

\section{Chosen-plaintext attack of DRPE with NLSE}
\label{SEC:CPA}

Chosen-plaintext attack (CPA) refers to the setting in which the attacker has access to the encryption oracle. Mathematically, this means the attacker has the information encoded in the encryption map:
\begin{equation}\label{EQ:CPA Data}
    f(\bx) \;\longmapsto\;
    g(\bx):=e^{i\varphi_2(\bx)}     \Lambda_{k,\beta}^{{L_z}}\!\Big [ f(\bx)e^{i\varphi_1(\bx)}\Big ]\,.
\end{equation}
In other words, the adversary may query~\eqref{EQ:CPA Data} with arbitrary plaintexts and observe the corresponding ciphertexts. From such information, the adversary intends to decode the encryption machine by recovering the encryption keys $\varphi_1$ and $\varphi_2$. Therefore, the CPA process can be formalized as follows.\\[2ex]
{\bf CPA Inverse Problem:} to reconstruct ($\varphi_1,\ \varphi_2,\ \beta$) from data encoded in the map~\eqref{EQ:CPA Data}.\\[2ex]
This is certainly one of the most informative attacks among all possible ones.

Before turning to detailed studies of the CPA inverse problem, let us mention that inverse problems for different variants of nonlinear Schr\"{o}dinger equations have been extensively studied; see, for instance, ~\cite{LaLuZh-arXiv23,LaUhYa-arXiv24,LaOkSaSaTe-arXiv24,LiLi-CTP21,ToBa-IEEE03,Ton-AAA02,WaLi-arXiv21,YaMu-JMS99} and references therein.

\subsection{Gauge invariance}
 
Let us first recall a basic structural nonuniqueness caused by an invariance of the nonlinear Schr\"{o}dinger model~\eqref{EQ:CNLS}. The equation is invariant under multiplication of the solution by a constant phase factor, i.e. if $u$ solves~\eqref{EQ:CNLS}, then so does $u e^{i\varphi_c}$ for any constant $\varphi_c\in\mathbb{R}$. This induces an intrinsic ambiguity in the encryption map~\eqref{EQ:CPA Data}.

\begin{proposition}[Gauge invariance]\label{prop:gauge}
For any constant $\varphi_c\in\mathbb{R}$, the encryption map~\eqref{EQ:CPA Data} satisfies
\[
e^{i \varphi_2(\bx)} \Lambda_{k,\beta}^{{L_z}}\!\Big [f e^{i\varphi_1}\Big]
\;=\;
e^{i(\varphi_2(\bx)-\varphi_c)}\,
\Lambda_{k,\beta}^{{L_z}}\!\Big [f\, e^{i(\varphi_1+\varphi_c)}\Big ]
\]
for all plaintexts $f$. Consequently, the pair of masks $(\varphi_1,\varphi_2)$ is identifiable at most up to an additive constant phase:
\[
(\varphi_1,\varphi_2)
\sim
(\varphi_1+\varphi_c,\;\varphi_2-\varphi_c).
\]
\end{proposition}

This invariance will be illustrated in the numerical reconstructions of~\Cref{SEC:Numer}, where the recovered masks agree with the true masks up to an additive constant.

\subsection{CPA with differential data}
\label{SEC:Linearization}

Existing strategies of CPA attacks on DRPE require the linearity of the encryption process (that is, they require the map $f\mapsto g$ to be linear)
~\cite{SiZh-OL04, CaMoArJu-OL05, PeWeZh-OL06}. For the nonlinear propagation model we have here, a standard strategy is to use only the information encoded in $f\mapsto |g|$ instead of the original map $f\mapsto g$. Since $|g|$ only depends on $\varphi_1$, $\varphi_2$ is eliminated temporarily, and we have a \emph{nonlinear phase retrieval} problem of recovering $\varphi_1$; see~\cite{ChReSo-IP25} for a similar phase retrieval problem. Once we know $\varphi_1$, $\varphi_2$ is readily reconstructed. Numerical simulations following this strategy will be provided in~\Cref{SEC:Numer}.
 
Here, we propose a differential CPA strategy to linearize the nonlinear propagation model. In this strategy, the attacker submits plaintexts that are closely related to each other and observes the difference in the resulting ciphertexts. The differential data are then used in the attacks. 

Let us consider perturbed inputs of the form
\begin{equation}\label{EQ:CPA-perturbed-input}
    f_\varepsilon(\bx)
    = \varepsilon f^{(1)}(\bx)
       + \tfrac{1}{2}\varepsilon^2 f^{(2)}(\bx)
       + \tfrac{1}{6}\varepsilon^3 f^{(3)}(\bx),
\end{equation}
and let $u_\varepsilon$ solve the NLSE with initial data $f_\varepsilon(\bx) e^{i\varphi_1(\bx)}$:
\begin{equation}\label{EQ:NLS eps}
    \begin{cases}
	i\,\partial_z u_\varepsilon
        +\dfrac{1}{2k}\Delta u_\varepsilon
        + \beta\dfrac{|u_\varepsilon|^2}{1+|u_\varepsilon|^2}\, u_\varepsilon = 0,
        & (z,\bx)\in (0, L_z)\times\Omega,\\[0.8ex]
	u_\varepsilon(0,\bx)=f_\varepsilon(\bx)e^{i\varphi_1(\bx)},
        & \bx\in\Omega.
    \end{cases}
\end{equation}
We expand $u_\varepsilon$ in powers of $\varepsilon$:
\[
    u_\varepsilon(z,\bx)
    = \varepsilon u^{(1)}(z,\bx)
      + \tfrac{1}{2}\varepsilon^2 u^{(2)}(z,\bx)
      + \tfrac{1}{6}\varepsilon^3 u^{(3)}(z,\bx)
      + \mathcal{O}(\varepsilon^4),
\]
where $u^{(j)} := \dfrac{d^j u_\varepsilon}{d\varepsilon^j}\big|_{\varepsilon=0}$. Substituting this expansion into \eqref{EQ:NLS eps} and equating coefficients of $\varepsilon$ yields a hierarchy of linear equations for the $u^{(j)}$. We summarize this structure in the following proposition.

\begin{proposition}
\label{prop:linearization}
Let $u_\varepsilon$ solve~\eqref{EQ:NLS eps} with initial data~\eqref{EQ:CPA-perturbed-input}, and let $u^{(j)}$ denote the $j$-th order derivatives with respect to $\varepsilon$ at $\varepsilon=0$ as above. Then:
\begin{enumerate}[label=(\roman*)]
    \item The first-order term $u^{(1)}$ solves the linear Schr\"odinger equation
    \begin{equation}\label{EQ:NLS Order eps}
        \begin{cases}
    	i\dfrac{\partial u^{(1)}}{\partial z} +\dfrac{1}{2k}\Delta u^{(1)} =0, & (z,\bx)\in (0, L_z)\times \Omega,\\[1.0ex]
    	u^{(1)}(0, \bx) =f^{(1)}(\bx)e^{i\varphi_1(\bx)},& \bx\in \Omega.
    	\end{cases}
    \end{equation}
    The corresponding first-order data are
    \begin{equation}\label{EQ:Data Order eps}
        g^{(1)}(\bx)
        := \frac{d g_\varepsilon}{d\varepsilon}(\bx)\Big|_{\varepsilon=0}
        = e^{i\varphi_2(\bx)}\, u^{(1)}(L_z,\bx).
    \end{equation}
    \item The second-order term $u^{(2)}$ also solves the same linear Schr\"odinger equation:
    \begin{equation}\label{EQ:NLS Order eps2}
        \begin{cases}
    	i\dfrac{\partial u^{(2)}}{\partial z} + \dfrac{1}{2k}\Delta u^{(2)} =0, & (z,\bx)\in (0, L_z)\times \Omega, \\[1.0ex]
    	u^{(2)}(0,\bx)=f^{(2)}(\bx)e^{i\varphi_1(\bx)},& \bx\in \Omega,
    	\end{cases}
    \end{equation}
    with second-order data
    \begin{equation}\label{EQ:Data Order eps2}
        g^{(2)}(\bx)
        := \frac{d^2 g_\varepsilon}{d\varepsilon^2}(\bx)\Big|_{\varepsilon=0}
        = e^{i\varphi_2(\bx)}\, u^{(2)}(L_z, \bx).
    \end{equation}
    Thus, up to second order in $\varepsilon$, the nonlinear term does not contribute.
    \item The third-order term $u^{(3)}$ satisfies
    \begin{equation}\label{EQ:NLS Order eps3}
        \begin{cases}
    	i\dfrac{\partial u^{(3)}}{\partial z} + \dfrac{1}{2k}\Delta u^{(3)} =-6\beta \big|u^{(1)}\big|^2u^{(1)}, & (z,\bx)\in (0, L_z)\times \Omega, \\[1.0ex]
    	u^{(3)}(0,\bx)=f^{(3)}(\bx)e^{i\varphi_1(\bx)},& \bx\in \Omega,
    	\end{cases}
    \end{equation}
    and the third-order data are
    \begin{equation}\label{EQ:Data Order eps3}
        g^{(3)}(\bx)
        := \frac{d^3 g_\varepsilon}{d\varepsilon^3}(\bx)\Big|_{\varepsilon=0}
        = e^{i\varphi_2(\bx)}\, u^{(3)}(L_z, \bx).
    \end{equation}
\end{enumerate}
\end{proposition}
The proof of this proposition follows the standard multilinearization procedure, originally proposed in~\cite{Isakov-ARMA93}, used in the analysis of
inverse problems for nonlinear PDEs, such as those in~\cite{LaLuZh-arXiv23,LaOkSaSaTe-arXiv24,ReSoWa-JDE25}. With the smoothness of the nonlinearity term, mainly the map $r\mapsto \frac{r}{1+r}$, near $r=0$, we only need to show the Fr\'echet differentiability of the map $f\mapsto u$ (in appropriate topology) at $f\equiv 0$. We omit the proof and refer interested readers to~\cite{LaLuZh-arXiv23,LaOkSaSaTe-arXiv24} for detailed discussions on a similar nonlinear Schr\"{o}dinger model.

Proposition~\ref{prop:linearization} shows that the first and second Fr\'echet derivatives of the encryption map at the origin
coincide with those of linear Schr\"odinger propagation. The nonlinear coefficient $\beta$
enters only in the third derivative, through the source term on the right-hand side of \eqref{EQ:NLS Order eps3}. That is, the third-order term is the first one that ``feels'' the nonlinearity. Since $\varphi_1$ and $\varphi_2$ both appear in the first-order linearization, we conclude that the nonlinear propagation exhibits no genuinely nonlinear behavior when using differential data. 

\section{Reconstructing phase masks using first-order data}
\label{sec:pointwise_attack}

In the linearized regime of Proposition~\ref{prop:linearization}, propagation is governed by the linear Schr\"odinger equation. For the case where the refractive index $n(\bx)$ (and therefore $k(\bx)$) is constant, we show that this corresponds exactly to the Fresnel transform, reducing the linearized nonlinear DRPE to the well-studied Fresnel domain DRPE of Situ and Zhang~\cite{SiZh-OL04}.

\begin{definition}[Fresnel DRPE~\cite{SiZh-OL04}]\label{def:Fresnel-DRPE}
Let $\mathrm{FrT}_{n,\lambda}[\cdot; z]$ denote the Fresnel transform over distance $z$ at wavelength $\lambda$:
\begin{equation}\label{EQ:FrT}
    \mathrm{FrT}_{n,\lambda}[u; z](\bx) := \frac{n}{i\lambda z} \int_{\mathbb{R}^d} \exp\!\left(\frac{in\pi}{\lambda z}|\bx - \by|^2\right) u(\by)\, d\by.
\end{equation}
The Fresnel-domain DRPE encrypts a plaintext $f$ as
\begin{equation}\label{EQ:DRPE-Fresnel}
    g_{\mathrm{Fresnel}}(\bx) =  \mathrm{FrT}_{n,\lambda}\big[e^{i\varphi_2(\bx)} \, \mathrm{FrT}_{n,\lambda}\big[ f(\bx) e^{i\varphi_1(\bx)}; {L_{z,1}} \big]; L_{z,2}\big],
\end{equation}
where $\varphi_1, \varphi_2$ are the phase mask keys and ${L_{z,1}},L_{z,2}$ are the propagation distances.
\end{definition}
\begin{remark}
    The outer Fresnel transform in Definition~\ref{def:Fresnel-DRPE} uses no secret parameters beyond $L_{z,2}$ and $\lambda$, which can be determined by the optical setup. An attacker with access to the complex ciphertext can invert this step directly. We therefore analyze the reduced scheme
    \begin{equation}\label{EQ:DRPE-Fresnel-simplified}
    \tilde{g}_{\mathrm{Fresnel}}(\bx) = e^{i\varphi_2(\bx)} \, \mathrm{FrT}_{n,\lambda}\big[ f(\bx) e^{i\varphi_1(\bx)}; {L_{z,1}} \big]
    \end{equation}
    without loss of generality. This simplification is analogous to Remark~\ref{REM:simplification} for nonlinear DRPE.
\end{remark}

\begin{proposition}[Equivalence of Linear Schr\"odinger DRPE and Fresnel DRPE]\label{prop:LS-vs-Fresnel}
Let $k(\bx)$ be constant and $\Omega = \mathbb{R}^d$. The linear Schr\"odinger propagator $\Lambda_k^{L_z}$ defined by
\begin{equation}\label{EQ:LS}
    i\partial_z u + \frac{1}{2k}\Delta u = 0, \qquad u(0,\bx) = u_0(\bx),
\end{equation}
coincides with the Fresnel transform: $\Lambda_k^{L_z} = \mathrm{FrT}_{n,\lambda}(\cdot; L_z)$ under the identification $k = 2\pi n/\lambda$. Consequently, the first-order ciphertext~\eqref{EQ:Data Order eps} of the linearized nonlinear DRPE is structurally identical to Fresnel DRPE:
\begin{equation}\label{EQ:equivalence-LS-Fresnel}
    g^{(1)}(\bx) = e^{i\varphi_2(\bx)} \Lambda_k^{L_z}\big[ f(\bx) e^{i\varphi_1(\bx)} \big] = e^{i\varphi_2(\bx)} \, \mathrm{FrT}_{n,\lambda}\big[ f e^{i\varphi_1}; L_z \big](\bx).
\end{equation}
\end{proposition}

\begin{proof}
The Fresnel transform~\eqref{EQ:FrT} acts as a Fourier multiplier. Writing $h_z(\bx) := \frac{n}{i\lambda z} \exp\!\left(\frac{in\pi}{\lambda z}|\bx|^2\right)$, we have
\[
    \mathrm{FrT}_{n,\lambda}[u; z] = \cF^{-1}(\mathcal{F}(u) \cdot \mathcal{F}(h_z)), 
    \qquad 
    \cF(h_z)(\bxi) = \exp\!\left(-\frac{i\lambda z}{4\pi n}|\bxi|^2\right).
\]
On the other hand, Fourier transforming~\eqref{EQ:LS} in $\bx$ yields the ODE $i\,\partial_z \widehat{u} - \frac{|\bxi|^2}{2k}\widehat{u} = 0$, with solution 
\[
    \widehat{u}(z,\bxi) = \exp\!\left(-i\frac{z}{2k}|\bxi|^2\right) \widehat{u_0}(\bxi).
\]
Under the identification $k = 2\pi n/\lambda$, the multiplier becomes $\exp\!\left(-i\frac{\lambda z}{4\pi n}|\bxi|^2\right) = \cF(h_z)(\bxi)$, which completes the proof.
\end{proof}

The Fresnel DRPE scheme and its vulnerabilities have been extensively studied~\cite{SiZh-OL04, CaMoArJu-OL05, PeWeZh-OL06}. In particular, since the map $f \mapsto g^{(1)}$ is a linear integral operator, it is susceptible to impulse-response profiling attacks. We make this precise in the next subsection.

The equivalence established in Proposition~\ref{prop:LS-vs-Fresnel} implies that standard cryptanalytic techniques for linear DRPE apply to the linearized nonlinear system. We describe a pointwise attack using delta-function probes in the flavor of~\cite{CaMoArJu-OL05}.

\begin{proposition}[Pointwise Phase Retrieval]\label{prop:pointwise}
Let the refractive index $n(\bx)$ be constant and $\Omega = \mathbb{R}^d$. Given access to an encryption oracle returning the complex ciphertext, $\varphi_2$ is recoverable from a single delta-function query, and each point value of $\varphi_1$ requires one additional query. Both masks are determined up to a shared global constant.
\end{proposition}
\begin{proof}
Querying with $f^{(1)}(\by) = \delta(\by - \wt\bx_1)$ and applying the Fresnel integral~\eqref{EQ:FrT} yields
\[
    g^{(1)}(\bx) = \left(\frac{n}{i\lambda L_z}\right)^{d/2} \exp\left(i\frac{n\pi}{\lambda L_z}|\bx - \wt\bx_1|^2\right) e^{i\varphi_1(\wt\bx_1)} e^{i\varphi_2(\bx)}.
\]
Since the Fresnel phase factor is determined by the optical parameters, a single probe recovers $\varphi_1(\wt\bx_1) + \varphi_2(\bx)$ from $\arg g^{(1)}(\bx)$. Because $\varphi_1(\wt\bx_1)$ is a constant (independent of $\bx$), this immediately determines $\varphi_2$ up to a global constant.

To recover $\varphi_1$, we query at a second point $\wt\bx_j \neq \wt\bx_1$. The ratio
\[
    \frac{g_j^{(1)}(\bx)}{g_1^{(1)}(\bx)} = \exp\left(i\frac{n\pi}{\lambda L_z}(|\bx - \wt\bx_j|^2 - |\bx - \wt\bx_1|^2)\right) e^{i(\varphi_1(\wt\bx_j) - \varphi_1(\wt\bx_1))}
\]
eliminates $\varphi_2(\bx)$. Compensating for the known Fresnel phase recovers $\varphi_1(\wt\bx_j) - \varphi_1(\wt\bx_1)$. Fixing $\wt\bx_1$ as a reference and varying $\wt\bx_j$ reconstructs $\varphi_1$ pointwise up to the same global constant.
\end{proof}
The global phase ambiguity $(\varphi_1, \varphi_2) \sim (\varphi_1 + c, \varphi_2 - c)$ is intrinsic to the system (Proposition~\ref{prop:gauge}) and does not affect decryption.

\section{Recovering nonlinearity using third-order data}
\label{SEC:Phase-Medium}

In the original nonlinear DRPE scheme~\cite{HoSi-EL22}, only the phase masks $(\varphi_1,\varphi_2)$ serve as encryption keys, while the physical parameters $k$ and $\beta$ are treated as known parameters determined by the optical setup. From a cryptographic standpoint, however, it is tempting to also make the physical parameters of the device part of the security keys. Here, we show that, on the mathematical level, adding $\beta$ as a security key does not provide additional security. In fact, $\beta$ can be reconstructed from the third-order data in the differential CPA attack strategy in~\Cref{SEC:Linearization}. 

\begin{theorem}[Uniqueness of Reconstructing $\beta$]
Assume that $L_z\neq \dfrac{mk}{|\xi|^2\pi}$ $\forall m\in\bbZ, \xi\in\bbZ^d$. Let $\Lambda_\beta$ denote the third-order effect $f^{(1)} \mapsto g^{(3)}$ derived from the multilinearization of the nonlinear Schr\"odinger equation. For any two coefficients $\beta_1, \beta_2 \in L^\infty(\Omega)$, if $\Lambda_{\beta_1} = \Lambda_{\beta_2}$ for all admissible inputs, then $\beta_1 = \beta_2$ almost everywhere.
\end{theorem}

\begin{proof}
With the notation of Section~\ref{SEC:Linearization}, the third-order perturbation $u_j^{(3)}$ satisfies, for $j=1,2$,
\begin{equation}\label{EQ:u(3)}
    \begin{cases}
    i \partial_z u_j^{(3)} + \dfrac{1}{2k} \Delta u_j^{(3)}
    = -6\beta_j |u^{(1)}|^2 u^{(1)}, & (z,\bx)\in (0,L_z)\times\Omega,\\[0.5ex]
    u_j^{(3)}(0,\bx) = f^{(3)}(\bx)e^{i\varphi_1(\bx)}.
    \end{cases}
\end{equation}
Let us choose $f^{(3)}\equiv 0$. Let $g_j^{(3)}$ denote the corresponding terminal data,
\[
    g_j^{(3)}(\bx)=e^{i\varphi_2(\bx)} u_j^{(3)}(L_z,\bx).
\]
To derive an orthogonality identity, let $v$ be any function satisfying the adjoint Schr\"odinger equation
\begin{equation}\label{EQ:v}
    -i\partial_z v + \dfrac{1}{2k}\Delta v = 0, \qquad (z,\bx)\in (0,L_z)\times\Omega.
\end{equation}
Multiplying~\eqref{EQ:u(3)} by $v$, multiplying~\eqref{EQ:v} by $u_j^{(3)}$, subtracting the resulting equations, and integrating over $(0,L_z)\times\Omega$ with the help of integration by parts and the periodic boundary conditions yields, for $j=1,2$,
\begin{equation}
    -6\int_0^{L_z}\int_{\Omega} \beta_j |u^{(1)}|^2 u^{(1)}v\,d\bx\,dz = i\int_{\Omega} v(L_z,\bx) u_j^{(3)}(L_z,\bx)\,d\bx.
\end{equation}
Since $g_1^{(3)} = g_2^{(3)}$ by assumption, we have $u_1^{(3)}(L_z,\bx) = u_2^{(3)}(L_z,\bx)$. Therefore, subtracting the above equation for $j=2$ from the equation for $j=1$ results in the orthogonality relation 
\begin{equation}\label{EQ:orthogonality-beta}
    \int_0^{L_z}\int_{\Omega} (\beta_1 - \beta_2) |u^{(1)}|^2 u^{(1)}v\,d\bx\,dz = 0.
\end{equation}

Without loss of generality, let $\Omega = [0,1]^d$ be the unit hypercube. To probe the frequency content of $\beta_1 - \beta_2$, we choose the plaintext probe $f^{(1)}$ as
\[
    f^{(1)}(\bx) = e^{-i\varphi_1(\bx)}e^{2\pi i\xi\cdot\bx},
\]
where $\xi\in \bbZ^d$. This choice generates the plane-wave solution
\[
    u^{(1)}(z,\bx) = \exp\left[i\left(2\pi\xi\cdot\bx - \frac{2\pi^2|\xi|^2}{k} z\right)\right].
\]
Similarly, for $v$, we choose the plane-wave solution 
\[
v(z,\bx) = \exp\left[i\left(2\pi\eta\cdot\bx + \frac{2\pi^2|\eta|^2}{k}z\right)\right].
\]
Substituting these into~\eqref{EQ:orthogonality-beta} gives
\begin{align*}
    \int_0^{L_z}\int_{\Omega} (\beta_1 - \beta_2)(\bx) \exp\left[i\left(2\pi(\xi + \eta)\cdot\bx + \frac{2\pi^2(|\eta|^2 - |\xi|^2)}{k} z\right)\right]\,d\bx\,dz = 0.
\end{align*}
We, therefore, have that the Fourier transform of $\beta_1 - \beta_2$ satisfies
\begin{align*}
    \widehat{\beta_1-\beta_2}(-(\xi+\eta))\int_0^{L_z}\exp\left[i\frac{2\pi^2(|\eta|^2-|\xi|^2)}{k}z\right]\,dz = 0,
\end{align*}
for all $\xi,\eta\in\bbZ^d$. Let us now take $\eta=-\xi+\xi_\perp$ (which is always possible when $d\ge 2$). Then the above relation is simplified to
\begin{align*}
    \widehat{\beta_1-\beta_2}(-\xi_\perp)\int_0^{L_z}\exp\left[i\frac{2\pi^2(|\xi_\perp|^2)}{k}z\right]\,dz = 0\,.
\end{align*}
With the assumption that $L_z\neq \dfrac{mk}{|\xi_\perp|^2\pi}$ $\forall m\in\bbZ, \xi_\perp\in\bbZ^d$, the integral in the above relation is nonzero. Therefore, we have that the Fourier transform of $\beta_1-\beta_2$ vanishes. Plancherel's theorem then implies that $\beta_1=\beta_2$ almost everywhere.
\end{proof}

Next, we develop an integral equation formulation for the computational reconstruction of $\beta$. We apply a similar delta-probe strategy: querying the system with $f^{(1)}(\bx) = \delta(\bx - \wt\bx)$, $f^{(3)}(\bx) = 0$, and analyzing the response.
\begin{proposition}[Integral Equation for $\beta$]\label{prop:beta-integral}
Let the refractive index $n(\bx)$ be constant and $\Omega = \mathbb{R}^d$. Let $f^{(1)}(\bx) = \delta(\bx - \wt\bx)$, $f^{(3)}(\bx) = 0$ be a delta-function probe. Then the third-order response $u^{(3)}(L_z, \bx)$ satisfies the linear integral equation
\begin{equation}\label{EQ:beta-integral}
    \int_{\bbR^d} \cK(\bx, \by; \wt\bx) \, \beta(\by) \, d\by = -u^{(3)}(L_z, \bx),
\end{equation}
where the kernel $\cK$ is given by
\begin{equation}\label{EQ:kernel-K}
    \cK(\bx, \by; \wt\bx) = \frac{6 n^{2d} e^{i\varphi_1(\wt\bx)}}{(i\lambda)^d \lambda^d} \int_0^{L_z} 
    \frac{1}{(L_z - z')^{d/2} (z')^{3d/2}} \,
    \exp\left(i\frac{n\pi}{\lambda(L_z - z')}|\bx - \by|^2 + i\frac{n\pi}{\lambda z'}|\by - \wt\bx|^2\right) dz'.
\end{equation}
\end{proposition}

\begin{proof}
From the third-order linearization~\eqref{EQ:NLS Order eps3}, the solution admits the Duhamel representation
\begin{equation}\label{EQ:duhamel-u3}
    u^{(3)}(z, \bx) = \mathrm{FrT}_{n,\lambda}\big[f^{(3)} e^{i\varphi_1}; z\big](\bx)
    - 6 \int_0^z \mathrm{FrT}_{n,\lambda}\big[\beta \, |u^{(1)}|^2 u^{(1)}; z - z'\big](\bx) \, dz'.
\end{equation}
We have chosen a probe with $f^{(3)}(\bx) = 0$, so the first term vanishes. The first-order field is
\[
    u^{(1)}(z', \by) = \left(\frac{n}{i\lambda z'}\right)^{d/2} 
    \exp\left(i\frac{n\pi}{\lambda z'}|\by - \wt\bx|^2\right) e^{i\varphi_1(\wt\bx)},
\]
which yields $|u^{(1)}(z', \by)|^2 = (n/\lambda z')^d$. Substituting into the Duhamel integral~\eqref{EQ:duhamel-u3} and writing the Fresnel transform explicitly, we obtain
\[
    u^{(3)}(L_z, \bx) = -\int_{\bbR^d} \cK(\bx, \by; \wt\bx) \, \beta(\by) \, d\by,
\]
where
\begin{equation}\label{EQ:kernel-K-expanded}
    \cK(\bx, \by; \wt\bx) = 6 e^{i\varphi_1(\wt\bx)} \int_0^{L_z} 
    \left(\frac{n}{i\lambda(L_z - z')}\right)^{d/2} 
    \left(\frac{n}{i\lambda z'}\right)^{d/2} 
    \left(\frac{n}{\lambda z'}\right)^d 
    e^{i\frac{n\pi}{\lambda(L_z - z')}|\bx - \by|^2}
    e^{i\frac{n\pi}{\lambda z'}|\by - \wt\bx|^2} \, dz'.
\end{equation}
Formula~\eqref{EQ:kernel-K} is from simplifying~\eqref{EQ:kernel-K-expanded}.
\end{proof}

The kernel $\cK(\bx, \by; \wt\bx)$ has a natural physical interpretation: it represents a composition of Fresnel propagators from the probe location $\wt\bx$ through the interaction point $\by$ to the observation point $\bx$. Unlike for the phase masks, where delta probes yield pointwise recovery (Proposition~\ref{prop:pointwise}), here they only reduce the inverse problem to a linear integral equation. However, it is not clear whether the integral operator induced by $\cK(\bx, \by; \wt\bx)$ is injective for a fixed $\wt\bx$, so a single probe may not suffice to determine $\beta$. Nonetheless, in the CPA setting, the attacker may vary $\wt\bx$ freely, which gives hope that the integral equation is invertible for reconstruction.

\section{Robustness of decryption with learned security keys}
\label{SEC:Stability}

In this section, we characterize the impact of uncertainty in the recovered security keys on the accuracy of the decrypted plaintexts, which was explicitly raised and empirically examined in \cite{HoSi-EL22}. Specifically, we establish stability estimates showing that small errors in the estimated keys (both the physical medium parameters and the phase masks) result in bounded, small errors in the recovered image.

The decryption process recovers the plaintext $f$ from the ciphertext $g$ by reversing the encryption steps. More precisely, recall from~\Cref{SEC:DRPE} that the decryption process is given by the initial value problem:
\begin{equation}\label{EQ:decry_beta_1}
    \begin{array}{rcll}
    i\partial_z u + \dfrac{1}{2k}\Delta u + \beta \dfrac{|u|^2}{1+|u|^2}u &=& 0, & \mbox{in}\ \ (0,L_z)\times \Omega \\[2ex]
    u(L_z, \bx) &=& g(\bx)e^{-i\varphi_2(\bx)}, & \mbox{in}\ \ \Omega,
    \end{array}
\end{equation}
where the propagation is from $z=L_z$ to $z=0$. The plaintext is then recovered as $f(\bx) = u(0,\bx)e^{-i\varphi_1(\bx)}$.

\subsection{Stability with respect to the nonlinearity \texorpdfstring{$\beta$}{}}

We first address the stability with respect to the nonlinear coefficient $\beta$. We assume the phase masks $\varphi_1$ and $\varphi_2$ have been recovered sufficiently well, and we focus strictly on the error induced by the medium parameter $\beta$.

Let $u$ and $\tilde u$ denote the solutions to~\eqref{EQ:decry_beta_1} using the true parameter $\beta$ and an approximation $\tilde\beta$ of the true parameter, respectively. Both satisfy the same initial condition at $z=L_z$ (the ciphertext), but evolve under slightly different dynamics. The true decryption satisfies
\begin{equation}\label{EQ:decry_beta}
    \begin{array}{rcll}
    i\partial_z u + \dfrac{1}{2k}\Delta u + \beta \dfrac{|u|^2}{1+|u|^2}u &=& 0, & \mbox{in}\ \ (0,L_z)\times \Omega \\[2ex]
    u(L_z, \bx) &=& h(\bx), & \mbox{in}\ \ \Omega,
    \end{array}
\end{equation}
while the perturbed decryption satisfies
\begin{equation}\label{EQ:decry_beta_tilde}
    \begin{array}{rcll}
    i\partial_z \tilde{u} + \dfrac{1}{2k}\Delta \tilde{u} + \tilde{\beta} \dfrac{|\tilde{u}|^2}{1+|\tilde{u}|^2}\tilde{u} &=& 0, &  \mbox{in}\ \ (0,L_z)\times \Omega\\[2ex]
    \tilde{u}(L_z, \bx) &=& h(\bx), & \mbox{in}\ \ \Omega.
    \end{array}
\end{equation}
Here, we use the shorthand $h(\bx) := g(\bx)e^{-i\varphi_2(\bx)}$. In this subsection, we focus on the case when $\Omega=\bbR^d$ and $k(\bx)$ is a constant.

\begin{lemma}\label{lem:energy estimate}
    Let $\Omega = \mathbb{R}^d$, $\beta \in L^\infty(\mathbb{R}^d)$, $h\in L^2(\mathbb{R}^d)$, and $k(\bx)$ be a positive constant. Then for $u$ solving~\eqref{EQ:decry_beta}, we have the following energy estimate for all $z\in [0,L_z]$:
    \begin{equation}\label{eq:energy estimate}
    ||u(z,\cdot)||_{L^2(\mathbb{R}^d)} \leq C||h||_{L^2(\mathbb{R}^d)},
    \end{equation}
    for some constant $C>0$ depending on $\|\beta\|_{L^{\infty}{(\mathbb{R}^d)}}$ and $L_z$.
\end{lemma}
\begin{proof}
     Doing a transform $v(z,\bx) := u(L_z-2kz,\bx)$, we have
\begin{equation}
    \begin{array}{rcll}
    -i\partial_z v + \Delta v + 2k\beta \dfrac{|v|^2}{1+|v|^2}v &=& 0,  & \mbox{in}\ \ \left[0, \dfrac{L_z}{2k}\right]\times \Omega\\[2ex]
    v(0, \bx) &=& h(\bx), &\mbox{in}\ \ \Omega.
    \end{array}
\end{equation}
By Duhamel's representation, 
\begin{equation}\label{eq:duhamel}
v(z,\cdot) = e^{-iz\Delta}h - 2i k\int_{0}^{z}e^{-i(z-s)\Delta}\left[\beta\frac{|v|^2}{1+|v|^2}v(\cdot,s)\right]\,ds
\end{equation}
where $e^{-iz\Delta}$ is the free propagator defined by 
$$
(e^{-iz\Delta}h)(\bx) := (2\pi)^{-d}\int_{\mathbb{R}^d}e^{i\bx\cdot\bxi+iz|\bxi|^2}\widehat{h}(\bxi)\,d\bxi.
$$
Notice that $e^{-iz\Delta}h(\bx) = \widecheck{e^{iz|\bxi|^2}\widehat{h}(\bxi)}$, so Parseval's identity gives $||e^{-iz\Delta}h||_{L^2(\mathbb{R}^d)} =||h||_{L^2(\mathbb{R}^d)}$. For the same reason, $\|e^{-i(z-s)\Delta}\left[\beta\frac{|v|^2}{1+|v|^2}v(\cdot,s)\right]\|_{L^2(\bbR^d)} = \|\beta\frac{|v|^2}{1+|v|^2}v(\cdot,s)\|_{L^2(\bbR^d)}$. Squaring~\eqref{eq:duhamel}, integrating, and using Cauchy-Schwarz therefore gives
\begin{align*}
    \frac{1}{2}\int_{\mathbb{R}^d}|v(z,\bx)|^2 \,d\bx &\leq \int_{\mathbb{R}^d}|e^{-iz\Delta} h(\bx)|^2 \,d\bx + 4k^2\int_{\mathbb{R}^d}\left|\int_{0}^{z}e^{-i(z-s)\Delta}\left[\beta\frac{|v|^2}{1+|v|^2}v(\cdot,s)\right] \,ds\right|^2 \,d\bx\\
    &\leq\int_{\mathbb{R}^d}|h(\bx)|^2 \,d\bx+4k^2 z\int_{0}^{z}\int_{\mathbb{R}^d}\left|e^{-i(z-s)\Delta}\left[\beta\frac{|v|^2}{1+|v|^2}v(\cdot,s)\right]\right|^2 \,d\bx \,ds\\
    &=\int_{\mathbb{R}^d}|h(\bx)|^2 \,d\bx+4k^2 z\int_{0}^{z}\left\|e^{-i(z-s)\Delta}\left[\beta\frac{|v|^2}{1+|v|^2}v(\cdot,s)\right]\right\|^2_{L^2(\mathbb{R}^d)} \,ds\\
    &=\int_{\mathbb{R}^d}|h(\bx)|^2 \,d\bx+4k^2 z\int_{0}^{z}\left\|\beta\frac{|v|^2}{1+|v|^2}v(\cdot,s)\right\|^2_{L^2(\mathbb{R}^d)} \,ds\\
    &\leq \int_{\mathbb{R}^d}|h(\bx)|^2 \,d\bx+4k^2\|\beta\|_{L^{\infty}(\mathbb{R}^d)}^2 z\int_{0}^z\int_{\mathbb{R}^d}|v(s,\bx)|^2 \,d\bx \,ds.
\end{align*}
Using the fact that $0\leq z \leq \frac{L_z}{2k}$, it follows that
\begin{equation*}
    \int_{\mathbb{R}^d}|v(z,\bx)|^2 \,d\bx \le 2\int_{\mathbb{R}^d}|h(\bx)|^2 \,d\bx + 4kL_z \|\beta\|_{L^{\infty}(\mathbb{R}^d)}^2 \int_{0}^z\int_{\mathbb{R}^d}|v(s,\bx)|^2 \,d\bx \,ds.
\end{equation*}

\noindent The integral form of Gronwall's inequality then gives us
\begin{equation*}
    \int_{\mathbb{R}^d}|v(z,\bx)|^2 \,d\bx \leq 2\exp\left(4kL_z\|\beta\|_{L^{\infty}(\mathbb{R}^d) }^2 z\right)\int_{\mathbb{R}^d}|h(\bx)|^2 \,d\bx \leq 2\exp\left(2L_z^2\|\beta\|_{L^{\infty}(\mathbb{R}^d) }^2\right)\int_{\mathbb{R}^d}|h(\bx)|^2.
\end{equation*}
Transforming back to $u$, we obtain the desired estimate~\eqref{eq:energy estimate} with $C = \sqrt{2}\exp\left(L_z^2\|\beta\|_{L^{\infty}(\mathbb{R}^d) }^2\right)$.
\end{proof}

With this lemma in hand, we are ready to establish the following stability estimate for the recovery of $\beta$:
\begin{theorem}\label{thm:beta_stability}
    Let $\Omega=\bbR^d$, $\beta, \tilde{\beta} \in L^\infty(\bbR^d)$, and $k(\bx)$ be a positive constant. Moreover, assume that $\|\beta\|_{L^{\infty}(\bbR^d)}, \|\tilde{\beta}\|_{L^{\infty}(\bbR^d)} \le M$. For a fixed $h \in L^2(\bbR^d)$, let $u$ and $\tilde{u}$ denote the solutions to the decryption dynamics \eqref{EQ:decry_beta} and \eqref{EQ:decry_beta_tilde}, respectively, and let $f(\bx) = u(0,\bx)e^{-i\varphi_1(\bx)}$ and $\tilde{f}(\bx) = \tilde{u}(0,\bx)e^{-i\varphi_1(\bx)}$ be the corresponding decrypted plaintexts. Then there exists a constant $C > 0$, depending on parameters $M$, $L_z$, and $\|h\|_{L^2(\bbR^d)}$, such that
    \begin{equation}
        \|f - \tilde{f}\|_{L^2(\bbR^d)} \le C \|\beta - \tilde{\beta}\|_{L^\infty(\bbR^d)}.
    \end{equation}
\end{theorem}

\begin{proof}
    Let us denote $w := u - \tilde{u}$. Subtracting the two equations yields 
    \[
    i\partial_z w + \frac{1}{2k}\Delta w + \beta \frac{|u|^2}{1+|u|^2}u- \tilde{\beta} \frac{|\tilde u|^2}{1+|\tilde u|^2} \tilde u = 0.
    \]
    Rearranging yields
    \[
    i\partial_z w  = - \frac{1}{2k}\Delta w - \beta \left(\frac{|u|^2}{1+|u|^2}u - \frac{|\tilde u|^2}{1+|\tilde u|^2} \tilde u\right) - (\beta - \tilde \beta) \frac{|\tilde u|^2}{1+|\tilde u|^2} \tilde u.
    \]
    It follows that
    \begin{align*}
        \frac{d\|w(z,\cdot)\|_{L^2(\bbR^d)}^2}{dz} &= 2 \text{Re} \int_{\Omega} \bar w \partial_z w \, d\bx \\
        &=-2\text{Im} \int_{\Omega} \bar w \left[ \beta \left(\frac{|u|^2}{1+|u|^2}u - \frac{|\tilde u|^2}{1+|\tilde u|^2} \tilde u\right) + (\beta - \tilde \beta) \frac{|\tilde u|^2}{1+|\tilde u|^2} \tilde u \right] d\bx.
    \end{align*}
    Now, the saturating nonlinearity is clearly Lipschitz continuous, so we have
    \[
    \left\| \frac{|u|^2}{1+|u|^2}u - \frac{|\tilde u|^2}{1+|\tilde u|^2} \tilde u \right\|_{L^2(\bbR^d)} \leq L \|w\|_{L^2(\bbR^d)}
    \]
    for some universal constant $L > 0$. Therefore,
    \[
    \left|\frac{d\|w(z,\cdot)\|_{L^2(\bbR^d)}^2}{dz}\right| \leq  2 \|\beta\|_{L^\infty(\bbR^d)} L \|w\|_{L^2(\bbR^d)}^2 + 2 \|\beta - \tilde{\beta}\|_{L^\infty(\bbR^d)}  \sup_{z\in [0,L_z]}\|{\tilde{u}}(z,\cdot)\|_{L^2(\bbR^d)}\|w\|_{L^2(\bbR^d)}.
    \]
    Using Lemma~\ref{lem:energy estimate}, we have $\sup_{z\in [0,L_z]}\|{\tilde{u}}(z,\cdot)\|_{L^2(\bbR^d)} \le C||h||_{L^2(\bbR^d)}$. It follows that
    \[
    \left|\frac{d\|w(z,\cdot)\|_{L^2(\bbR^d)}}{dz}\right| \leq C_1\|w(z,\cdot)\|_{L^2(\bbR^d)} + C_2\|\beta - \tilde{\beta}\|_{L^\infty(\bbR^d)} 
    \]
    where $C_1 := 2M L$ and $C_2 := 2C||h||_{L^2(\bbR^d)}$. Applying Gronwall's inequality and using the fact that $w(L_z,\bx) = 0$, we have
    \[
    \|w(0,\cdot)\|_{L^2(\bbR^d)} \le L_z C_2 \|\beta - \tilde{\beta}\|_{L^\infty(\bbR^d)} e^{C_1 L_z}.
    \]
    Since $f(\bx) = u(0,\bx) e^{-i\varphi_1(\bx)}$, $\tilde{f}(\bx) = \tilde{u}(0,\bx) e^{-i\varphi_1(\bx)}$, and $|e^{-i\varphi_1(\bx)}| = 1$, we can set $C = L_z C_2 e^{C_1 L_z}$ and conclude that 
    \[
    \|f - \tilde{f}\|_{L^2(\bbR^d)} \le C \|\beta - \tilde{\beta}\|_{L^\infty(\bbR^d)}. \qedhere
    \]
\end{proof}

\subsection{Stability with respect to phase masks \texorpdfstring{$\varphi_1,\varphi_2$}{}}

We now address the stability of the decryption with respect to errors in the reconstructed phase masks. Let $(\varphi_{1}, \varphi_{2})$ be the true phase masks and $(\tilde{\varphi}_{1}, \tilde{\varphi}_{2})$ be the reconstructed masks.
Unlike the coefficient $\beta$, errors in $\varphi_{1}$ and $\varphi_{2}$ enter through the initial and terminal conditions. However, via a certain transformation, we can transfer these initial and terminal errors into the coefficients of the PDE, which allows us to utilize standard energy methods. 

Another caveat is that, from Proposition~\ref{prop:gauge} that phase masks are identifiable only up to a global constant: the pairs $(\varphi_1, \varphi_2)$ and $(\varphi_1 + c, \varphi_2 - c)$ yield identical experimental data. Therefore, as our Theorem~\ref{thm:phase_stability} will show, the stability is controlled by gauge-invariant functions of $(\varphi_1, \varphi_2)$ and $(\tilde{\varphi}_{1}, \tilde{\varphi}_{2})$.

We define a phase interpolation function $\Phi(z,\bx)$ by
\begin{equation}\label{eq:phase_interp}
    \Phi(z,\bx) = \left( 1 - \frac{z}{L_z} \right) \varphi_1(\bx) - \frac{z}{L_z}\varphi_2(\bx).
\end{equation}
The transformed function $U(z,\bx) := u(z,\bx) \, e^{-i\Phi(z,\bx)}$ satisfies mask-independent initial and terminal conditions:
\[
    U(0,\bx) = f(\bx) \quad \text{and} \quad U(L_z,\bx) = g(\bx).
\]
A direct calculation shows that $U$ satisfies the \emph{magnetic nonlinear Schr\"odinger equation}:
\begin{equation}\label{eq:magnetic_NLS}
    i\partial_z U + \frac{1}{2k}\Delta_{A} U + qU + \beta \frac{|U|^2}{1+|U|^2} U = 0,
\end{equation}
where the magnetic Laplacian is $\Delta_A := (\nabla + iA)^2$, and the potentials are
\begin{equation}\label{eq:mag_potentials}
    A(z,\bx) := \nabla_{\bx} \Phi(z,\bx), \qquad q(\bx) := -\partial_z \Phi(z,\bx) = \frac{\varphi_1(\bx) + \varphi_2(\bx)}{L_z}.
\end{equation}

\begin{lemma}\label{lem:potential_bounds}
    Let $B := A - \tilde{A}$ and $Q := q - \tilde{q}$. Then, for all $z \in [0, L_z]$:
    \begin{equation}\label{eq:B_bound}
        \|B(z,\cdot)\|_{L^\infty(\Omega)} \leq \|\nabla\varphi_1 - \nabla\tilde{\varphi}_1\|_{L^\infty(\Omega)} + \|\nabla\varphi_2 - \nabla\tilde{\varphi}_2\|_{L^\infty(\Omega)},
    \end{equation}
    \begin{align}\label{eq:q_bound}
        \|Q\|_{L^\infty(\Omega)} &\leq \frac{C_\Omega}{L_z} \left( \|\nabla\varphi_1 - \nabla\tilde{\varphi}_1\|_{L^\infty(\Omega)} + \|\nabla\varphi_2 - \nabla\tilde{\varphi}_2\|_{L^\infty(\Omega)} \right) +  \left|Q_\Omega \right|
    \end{align}
    where 
    \begin{equation}\label{eq:Q_Omega}
    Q_\Omega := \frac{1}{L_z|\Omega|}\int_\Omega \Big[(\varphi_1 - \tilde{\varphi}_1) + (\varphi_2 - \tilde{\varphi}_2)\Big]\,d\bx.
    \end{equation}
\end{lemma}

\begin{proof}
    From \eqref{eq:phase_interp}, the magnetic difference is
    \[
        B(z,\bx) = \left(1 - \frac{z}{L_z}\right)\nabla(\varphi_1 - \tilde{\varphi}_1) - \frac{z}{L_z}\nabla(\varphi_2 - \tilde{\varphi}_2).
    \]
    The bound \eqref{eq:B_bound} follows immediately from the triangle inequality and convexity.
    
    For the scalar potential, we have $Q = [(\varphi_1 - \tilde{\varphi}_1) + (\varphi_2 - \tilde{\varphi}_2)]/L_z$. By Proposition~\ref{prop:gauge}, $Q$ is gauge-invariant. Applying the Poincar\'e--Wirtinger inequality on the periodic domain $\Omega$ yields
    \[
        \left\|Q -  Q_\Omega \right\|_{L^\infty(\Omega)} \leq \frac{C_\Omega}{L_z} \|\nabla(\varphi_1 - \tilde{\varphi}_1) + \nabla(\varphi_2 - \tilde{\varphi}_2)\|_{L^\infty(\Omega)},
    \]
    where $Q_\Omega$ is defined in~\eqref{eq:Q_Omega}. Separating the mean term gives us 
    \[
    \left\|Q \right\|_{L^\infty(\Omega)} \leq \frac{C_\Omega}{L_z} \|\nabla(\varphi_1 - \tilde{\varphi}_1) + \nabla(\varphi_2 - \tilde{\varphi}_2)\|_{L^\infty(\Omega)} + \left|Q_\Omega\right|. 
    \qedhere
    \]
\end{proof}

\begin{theorem}[Stability with respect to phase masks]\label{thm:phase_stability}
    Let $(\varphi_1, \varphi_2)$ and $(\tilde{\varphi}_1, \tilde{\varphi}_2)$ be two pairs of phase masks in $W^{1,\infty}(\Omega)$. Moreover, assume that $\|\nabla\varphi_1\|_{L^\infty(\Omega)}, \|\nabla\varphi_2\|_{L^\infty(\Omega)} \le M$. For a given ciphertext $g \in L^\infty(\Omega)$, let $f$ and $\tilde{f}$ denote the corresponding decrypted plaintexts. Then there exists a constant $C > 0$, depending on $\Omega$, $M$, $k$, $\|\beta\|_{L^\infty(\Omega)}$, $L_z$, and $\|g\|_{L^\infty(\Omega)}$, such that
    \begin{equation}\label{eq:phase_stability}
        \|f - \tilde{f}\|_{L^2(\Omega)} \leq C \left( \|\nabla\varphi_1 - \nabla\tilde{\varphi}_1\|_{L^\infty(\Omega)} + \|\nabla\varphi_2 - \nabla\tilde{\varphi}_2\|_{L^\infty(\Omega)} + |Q_\Omega| \right).
    \end{equation}
\end{theorem}

\begin{proof}
    Let $U$ and $\tilde{U}$ be the gauge-transformed solutions satisfying \eqref{eq:magnetic_NLS} with potentials $(A, q)$ and $(\tilde{A}, \tilde{q})$, subject to the common terminal condition $U(L_z, \bx) = \tilde{U}(L_z, \bx) = g(\bx)$.
    
    Define the error $W := U - \tilde{U}$. Subtracting the equations for $U$ and $\tilde{U}$ yields
    \begin{equation}\label{eq:W_equation}
        i\partial_z W + \frac{1}{2k}\Delta_A W + qW + \beta [\mathcal{N}(U) - \mathcal{N}(\tilde{U})] = R,
    \end{equation}
    where $\mathcal{N}(u) := u|u|^2/(1+|u|^2)$ and the source term is $R := -\frac{1}{2k}(\Delta_A - \Delta_{\tilde{A}})\tilde{U} - (q - \tilde{q})\tilde{U}$.
    
    Expanding the difference of magnetic Laplacians:
    \[
        \Delta_A - \Delta_{\tilde{A}} = 2iB \cdot \nabla + i(\nabla \cdot B) - (|A|^2 - |\tilde{A}|^2).
    \]
    Using Lemma \ref{lem:potential_bounds} and standard regularity bounds on $\tilde{U}$, we estimate the source term:
    \begin{equation}\label{eq:R_bound}
        \|R(z,\cdot)\|_{L^2(\Omega)} \leq C_1 \left( \|\nabla\varphi_1 - \nabla\tilde{\varphi}_1\|_{L^\infty(\Omega)} + \|\nabla\varphi_2 - \nabla\tilde{\varphi}_2\|_{L^\infty(\Omega)} + |Q_\Omega| \right).
    \end{equation}
    
    Multiplying \eqref{eq:W_equation} by $\overline{W}$ and integrating over $\Omega$ yields
\begin{equation*}
    \int_\Omega i (\partial_z W) \overline{W} \, d\bx + \int_\Omega \frac{1}{2k}(\Delta_A W) \overline{W} \, d\bx + \int_\Omega q |W|^2 \, d\bx =  \int_\Omega \beta [\mathcal{N}(\tilde{U}) - \mathcal{N}(U)] \overline{W} \, d\bx + \int_\Omega R \overline{W} \, d\bx.
\end{equation*}
Taking the imaginary part of both sides, and the first term on the left-hand side becomes:
\begin{equation}\label{eq:est_time}
    \mathrm{Im} \int_\Omega i (\partial_z W) \overline{W} \, d\bx = \mathrm{Re} \int_\Omega (\partial_z W) \overline{W} \, d\bx = \frac{1}{2}\frac{d}{dz}\|W\|_{L^2(\Omega)}^2.
\end{equation}
The magnetic Laplacian is self-adjoint under periodic boundary conditions, and $q$ is real-valued, so the following imaginary parts vanish identically:
\begin{equation}\label{eq:est_linear}
    \mathrm{Im}\int_\Omega \frac{1}{2k}(\Delta_A W) \overline{W} \, d\bx = \mathrm{Im}\int_\Omega q |W|^2 \, d\bx = 0.
\end{equation}
For the right-hand side, the nonlinearity is Lipschitz continuous, so:
\begin{equation}\label{eq:est_nonlinear}
    \left| \mathrm{Im} \int_\Omega \beta [\mathcal{N}(\tilde{U}) - \mathcal{N}(U)] \overline{W} \, d\bx \right| 
    \leq L \|\beta\|_{L^\infty(\Omega)} \int_\Omega |W|^2 \, d\bx 
    =: C_\beta \|W\|_{L^2(\Omega)}^2.
\end{equation}
Finally, applying the Cauchy-Schwarz inequality to the source term gives:
\begin{equation}\label{eq:est_source}
    \left| \mathrm{Im} \int_\Omega R \overline{W} \, d\bx \right| \leq \|R(z,\cdot)\|_{L^2(\Omega)} \|W(z,\cdot)\|_{L^2(\Omega)}.
\end{equation}

Combining the time derivative \eqref{eq:est_time}, \eqref{eq:est_linear}, \eqref{eq:est_nonlinear} and \eqref{eq:est_source}, we have
    \[
        \frac{1}{2}\left|\frac{d}{dz}\|W\|_{L^2(\Omega)}^2\right| \leq C_\beta \|W\|_{L^2(\Omega)}^2 + \|R\|_{L^2(\Omega)}\|W\|_{L^2(\Omega)}.
    \]
    Using the identity $\frac{1}{2}\frac{d}{dz}\|W\|^2 = \|W\|\frac{d}{dz}\|W\|$, we simplify this to
    \[
        \left|\frac{d}{dz}\|W(z,\cdot)\|_{L^2(\Omega)}\right| \leq C_\beta \|W(z,\cdot)\|_{L^2(\Omega)} + \|R(z,\cdot)\|_{L^2(\Omega)}.
    \]
    Applying Gronwall's inequality and using the fact that $W(L_z,\bx) = 0$, we have
    \begin{align*}
        \|W(0,\cdot)\|_{L^2(\Omega)} &\leq e^{C_\beta L_z}\int_0^{L_z}  \|R(z,\cdot )\|_{L^2(\Omega)} \, dz \\
        &\leq e^{C_\beta L_z} L_z C_1 \left( \|\nabla\varphi_1 - \nabla\tilde{\varphi}_1\|_{L^\infty(\Omega)} + \|\nabla\varphi_2 - \nabla\tilde{\varphi}_2\|_{L^\infty(\Omega)} + |Q_\Omega| \right),
    \end{align*}
    upon substituting the bound \eqref{eq:R_bound} for $\|R\|_{L^2}$. This gives us the stability estimate \eqref{eq:phase_stability}.\qedhere
\end{proof}

\section{Numerical experiments}
\label{SEC:Numer}

We now present some numerical simulations to elaborate on the theoretical development in the previous sections. Instead of directly implementing the multilinearization differential data attack process, we perform numerical reconstruction based on a two-step procedure. In the first step, we reconstruct $\varphi_1$ and $\beta$, eliminating $\varphi_2$ by removing phase information in the ciphertext. This then becomes a phase retrieval problem. The second step recovers $\varphi_2$.
We assume that we have $N_s$ plaintext-ciphertext pairs $\{(f_s, g_s)\}_{s=1}^{N_s}$. Let $\{d_s\}_{s=1}^{N_s}$ denote the measured ciphertext amplitudes, i.e., $d_s(\bx) = |g_s(\bx)|$ where $g_s$ is produced by the encryption oracle using the true parameters. We perform the reconstruction by minimizing the least-squares intensity misfit:
\begin{equation}\label{EQ:OBJ Func}
    \Phi(\varphi_1, \beta) = \frac{1}{2} \sum_{s=1}^{N_s} \int_\Omega \big(|u_s(L_z,\bx)|^2 - d_s(\bx)^2\big)^2 \, d\bx,
\end{equation}
where $u_s$ satisfies the NLSE with plaintext $f_s$:
\begin{equation*}
    \begin{array}{rcll}
    i\dfrac{\partial u_s}{\partial z} + \dfrac{1}{2k}\Delta u_s + \beta(\bx)\dfrac{|u_s|^2}{1+|u_s|^2}u_s &=& 0, & \mbox{in}\ \ (0, L_z]\times \Omega\\[2ex]
    u_s(0, \bx)&=& f_s(\bx)e^{i\varphi_1(\bx)}, & \mbox{in}\ \ \Omega
    \end{array}
\end{equation*}
with periodic boundary conditions. The numerical implementation of the optimization problem follows that of~\cite{ChReSo-IP25} and is detailed in~\Cref{APP:Implementation}. Since we are using only a finite number of plaintext-ciphertext pairs in the attack, our numerical implementation is closer to the known–plaintext attack (KPA) situation~\cite{PeZhWeYu-OL06}.

\subsection{Leading-order effect: CPA of linear DRPE}
\label{sec:CPA}

We begin with the baseline CPA experiment corresponding to the first-order effect in the linearization analysis of~\Cref{SEC:Linearization}. The propagation model is~\eqref{EQ:NLS Order eps}. We performed simulations for different $k$ values, but will present only those for $k = 5$, as $k$ does not play an essential role. 

\paragraph{Numerical Experiment I.} The attacker queries the encryption oracle with the full set of $N_s=40$ sinusoidal plaintexts $\{f_s\}_{s=1}^{40}$ of the form
\[
    f_s(x,y) = 1 + 0.3 \sin(4s\pi x) + 0.3\sin(4s \pi y),\qquad s = 1,\dots,40.
\]
and reconstructs $\varphi_1$ via the phase retrieval, and subsequently extracts $\varphi_2$ from a single ciphertext. 

\begin{figure}[!htb]
    \centering
    \includegraphics[width=0.32\textwidth]{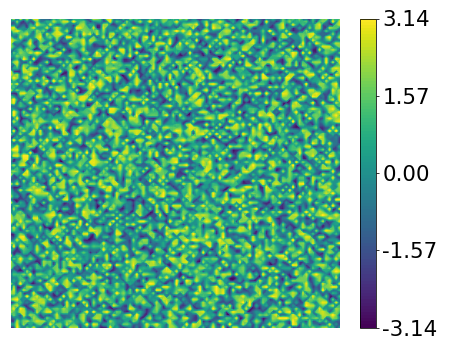}
    \includegraphics[width=0.32\textwidth]{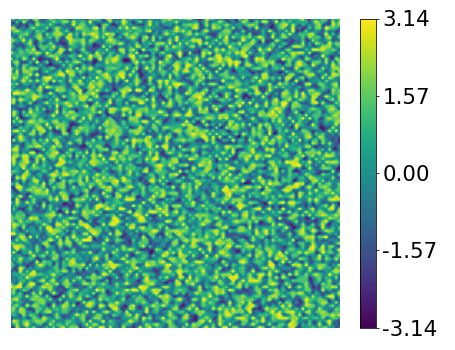}
    \includegraphics[width=0.32\textwidth]{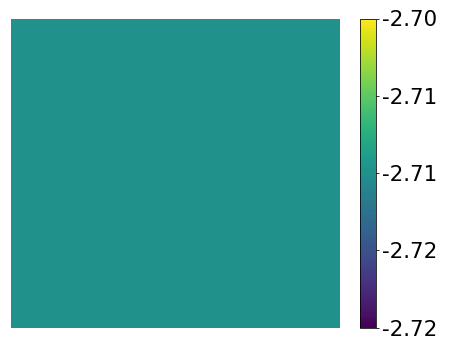}\\
    \includegraphics[width=0.32\textwidth]{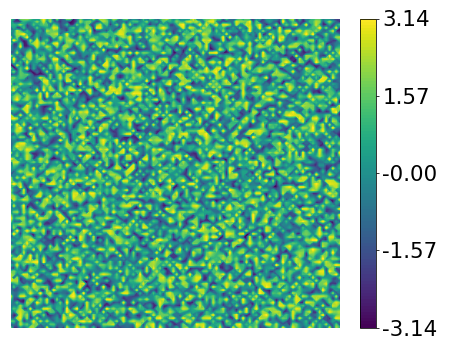}
    \includegraphics[width=0.32\textwidth]{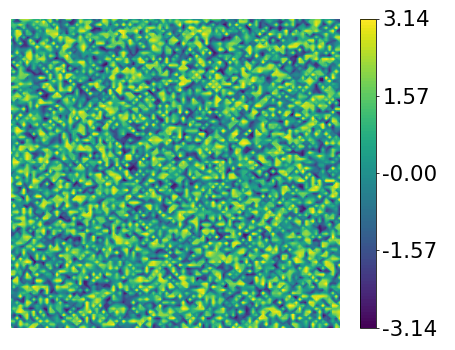}
    \includegraphics[width=0.32\textwidth]{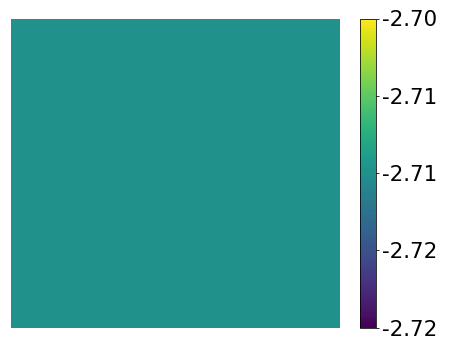}
    \caption{From left to right: true security keys $\begin{bmatrix}\varphi_1\\ \varphi_2\end{bmatrix}$ (left), reconstructed security keys $\begin{bmatrix}\wt \varphi_1\\ \wt \varphi_2\end{bmatrix}$(middle), and the errors in the reconstructions $\begin{bmatrix}\varphi_1-\wt\varphi_1\\ \varphi_2-\wt\varphi_2\end{bmatrix}$ (right).}
    \label{FIG:Ex1}
\end{figure}
In Figure~\ref{FIG:Ex1}, we show the true security keys $\varphi_1$ (top) and $\varphi_2$ (bottom), the reconstructions $\wt \varphi_1$ and $\wt\varphi_2$, as well as the error in the reconstructions. The reconstruction results exhibit excellent agreement with the ground truth. Both $\wt\varphi_1$ and $\wt \varphi_2$ match their corresponding masks up to the intrinsic global phase ambiguity, and the error plots show no spatial structure beyond numerical noise. 

In Figure~\ref{FIG:Ex1-B}, we show the impact of the reconstruction error on the decryption process. Shown from left to right are the true plaintext, the corresponding ciphertext using the keys $\varphi_1$ and $\varphi_2$, and the decrypted plaintext using the recovered security keys $\wt\varphi_1$ and $\wt\varphi_2$. The decrypted plaintext closely reproduces the original image, demonstrating that the recovered masks $(\tilde\varphi_1,\tilde\varphi_2)$ are fully operational.

\begin{figure}[!htb]
    \centering
    \includegraphics[width=0.32\textwidth]{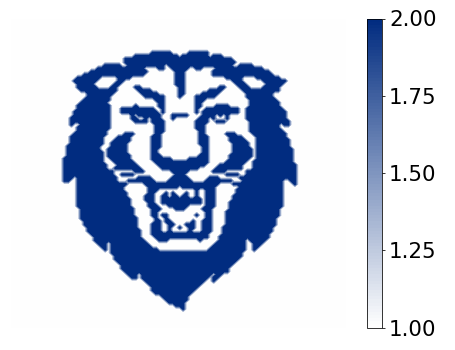}
    \includegraphics[width=0.32\textwidth]{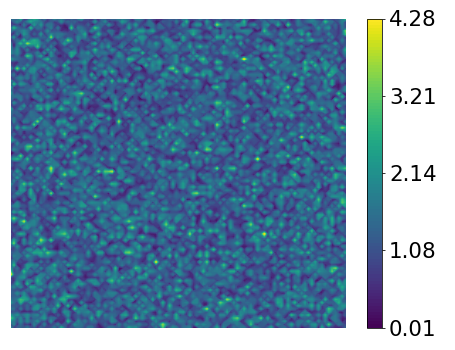}
    \includegraphics[width=0.32\textwidth]{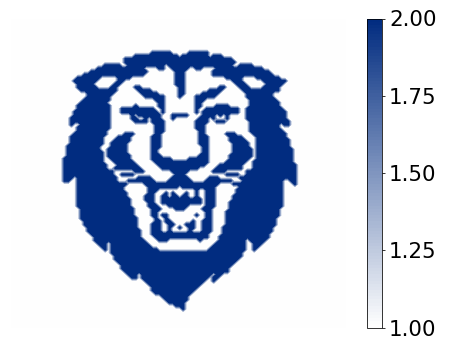}
    \caption{True plaintext (left), amplitude of the ciphertext (middle), decrypted plaintext (right).}
    \label{FIG:Ex1-B}
\end{figure}

\subsection{Joint recovery of \texorpdfstring{$(\varphi_1, \varphi_2, \beta)$}{}}
\label{sec:JointBeta}

We now present simulation results on the joint recovery of the security keys $\varphi_1$, $\varphi_2$, and the nonlinearity $\beta$. We use the same phase masks $\varphi_1$ and $\varphi_2$ as in Experiment I throughout all subsequent experiments; we have tested with other randomly generated masks and observed consistent results.

\paragraph{Numerical Experiment II.} Figure~\ref{FIG:Exp3_Joint} shows the first numerical test on the joint inversion. The same set of $N_s=40$ plaintext-ciphertext pairs is used in the reconstruction. The propagation distance is set to be $L_z = 0.01$. Aside from the obvious constant shift, the reconstructions are fairly accurate (see, for instance, the plots of the errors $\varphi_1-\wt\varphi_1$, $\varphi_2-\wt\varphi_2$, and $\beta-\wt\beta$). 
\begin{figure}[!htb]
    \centering
    \includegraphics[width=0.25\textwidth]{Figures/Exp1/phase_1_true.png}
    \includegraphics[width=0.25\textwidth]{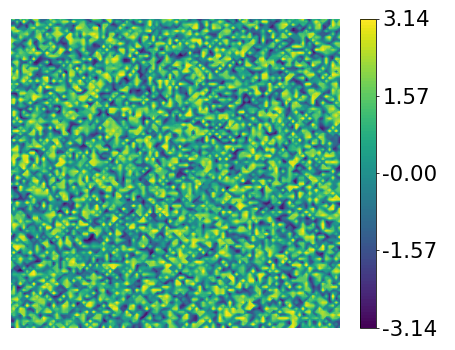}    
    \includegraphics[width=0.25\textwidth]{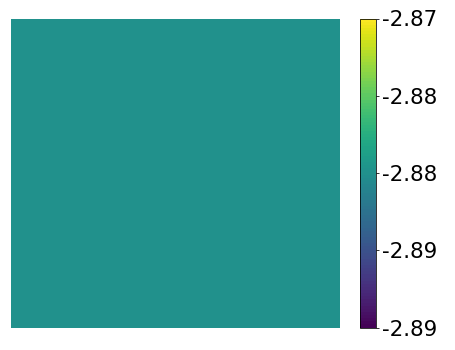}\\
    \includegraphics[width=0.25\textwidth]{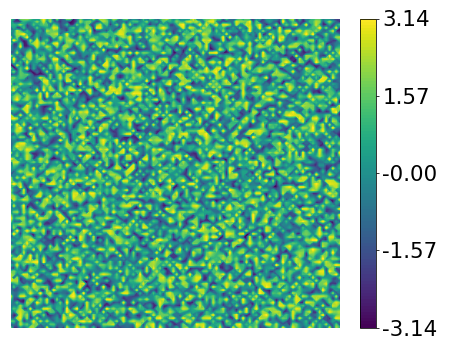}
    \includegraphics[width=0.25\textwidth]{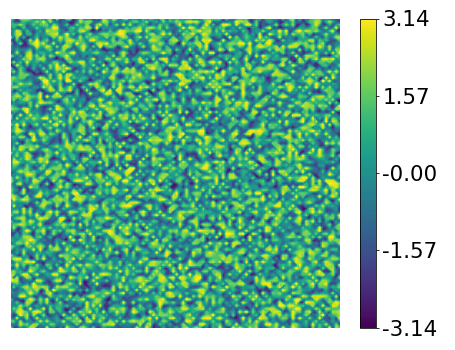}
    \includegraphics[width=0.25\textwidth]{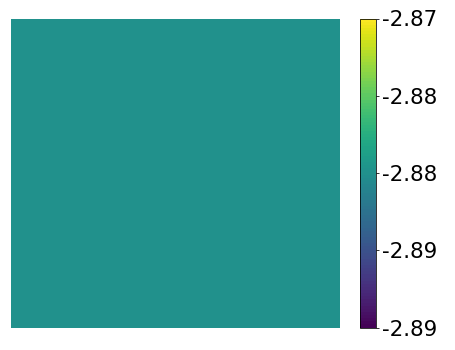}\\
    \includegraphics[width=0.25\textwidth]{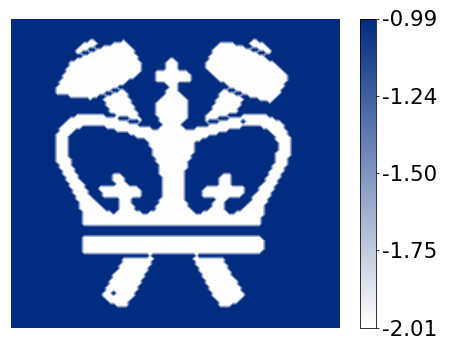}
    \includegraphics[width=0.25\textwidth]{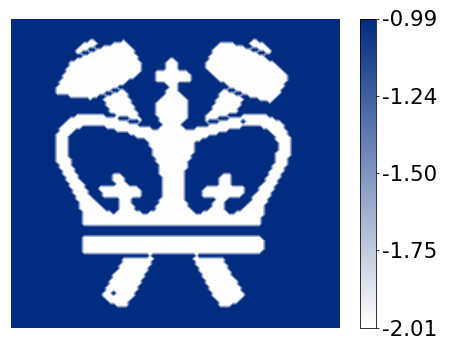}
    \includegraphics[width=0.25\textwidth]{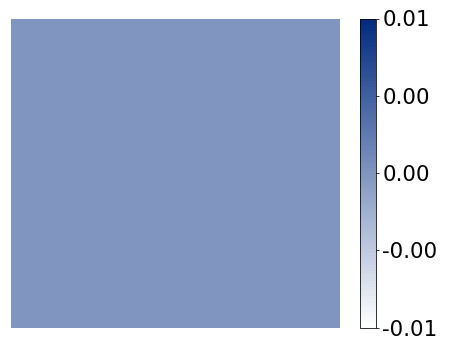}
    \caption{Joint recovery of $(\varphi_1, \varphi_2, \beta)$ in Numerical Experiment II. Shown are true $(\varphi_1, \varphi_2, \beta)$ (left), the reconstruction $(\wt\varphi_1, \wt\varphi_2, \wt\beta$) (middle), and the errors $(\varphi_1-\wt\varphi_1, \varphi_2-\wt\varphi_2, \beta-\wt\beta)$ (right). }
    \label{FIG:Exp3_Joint}
\end{figure}

We show in Figure~\ref{FIG:Exp3_Dec} a decryption result using the reconstructed security keys and the nonlinearity $(\wt\varphi_1,\wt\varphi_2, \wt\beta)$. The decryption quality is again very high. This numerical experiment shows that even when we include the nonlinearity $\beta$ as an additional security key, the system can remain vulnerable to CPA attack.
\begin{figure}
    \centering
    \includegraphics[width=0.3\textwidth]{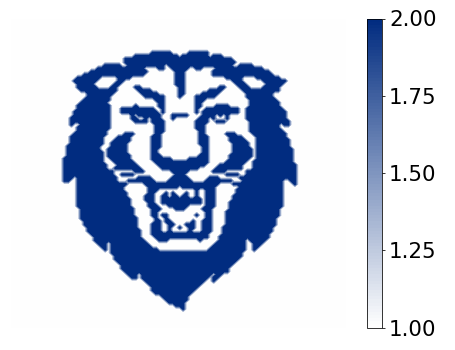}
    \includegraphics[width=0.3\textwidth]{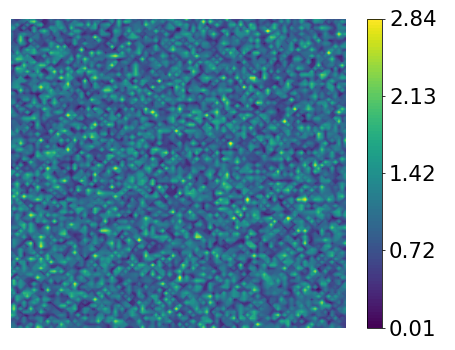}
    \includegraphics[width=0.3\textwidth]{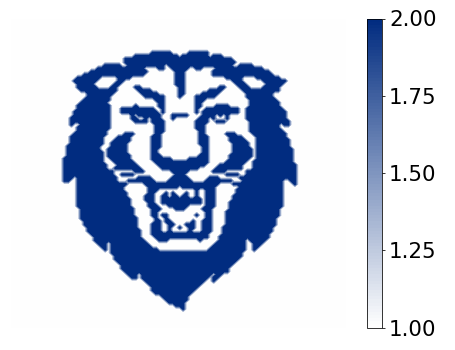}
    \caption{The true plaintext (left), amplitude of the ciphertext (middle), decrypted plaintext (right) using the reconstructed keys $(\wt\varphi_1, \wt\varphi_2, \wt\beta)$ in Numerical Experiment II.}
    \label{FIG:Exp3_Dec}
\end{figure}

\begin{figure}[!htb]
    \centering
    \includegraphics[width=0.25\textwidth]{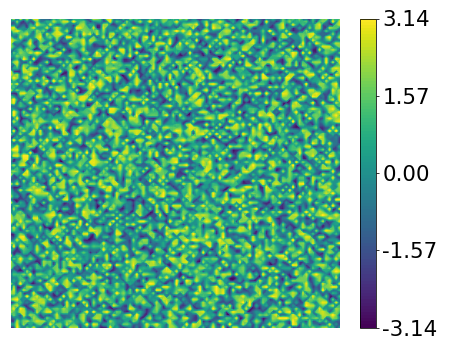}
    \includegraphics[width=0.25\textwidth]{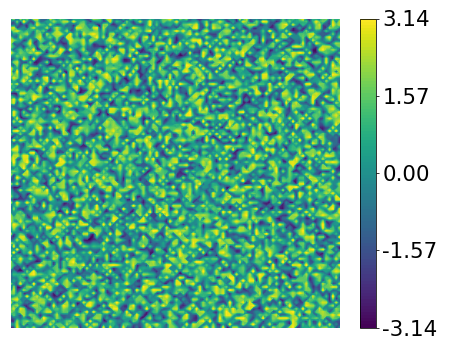}
    \includegraphics[width=0.25\textwidth]{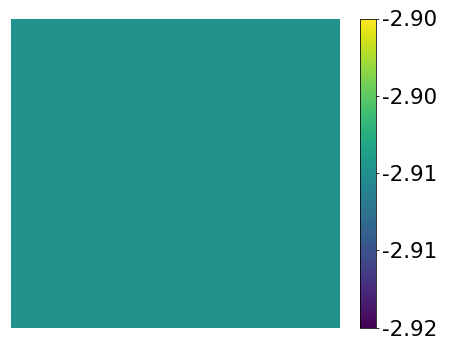}\\
    \includegraphics[width=0.25\textwidth]{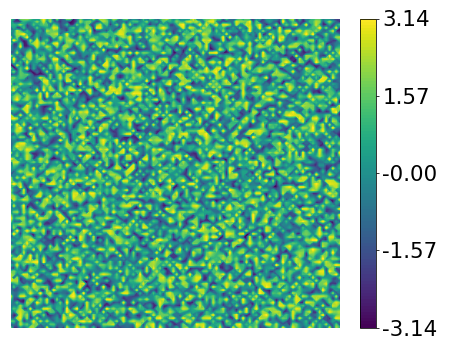}
    \includegraphics[width=0.25\textwidth]{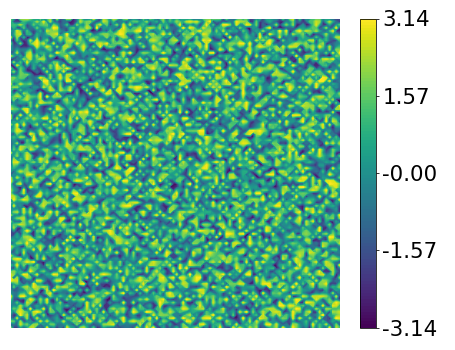}
    \includegraphics[width=0.25\textwidth]{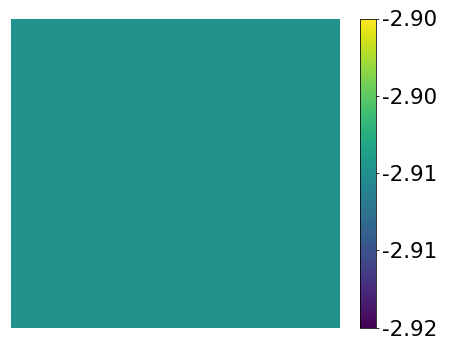}\\
    \includegraphics[width=0.25\textwidth]{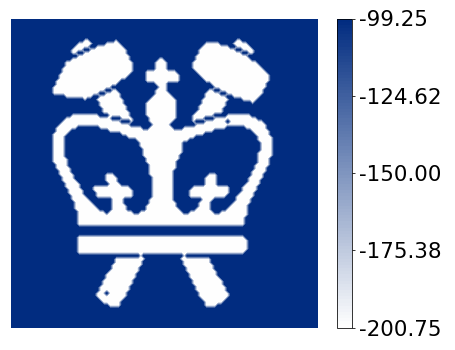}
    \includegraphics[width=0.25\textwidth]{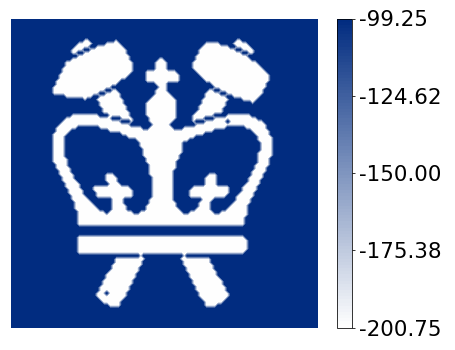}
    \includegraphics[width=0.25\textwidth]{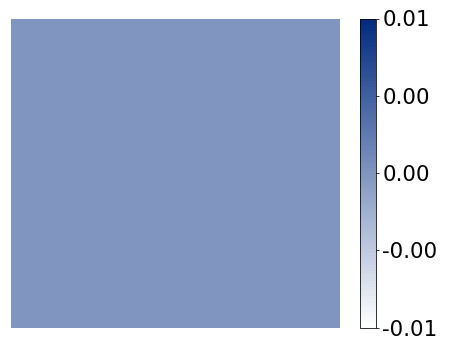}
    \caption{Joint recovery of $(\varphi_1, \varphi_2, \beta)$ in Numerical Experiment III. Shown are true $(\varphi_1, \varphi_2, \beta)$ (left), the reconstruction $(\wt\varphi_1, \wt\varphi_2, \wt\beta$) (middle), and the errors $(\varphi_1-\wt\varphi_1, \varphi_2-\wt\varphi_2, \beta-\wt\beta)$ (right). 
    }
    \label{FIG:JointBeta100Success}
\end{figure}
\paragraph{Numerical Experiment III.} In the next numerical experiment, we increase the intensity of the nonlinearity $\beta$ to $100$ times that of the previous numerical experiment. The joint reconstruction in this case also gives fairly accurate results; see Figure~\ref{FIG:JointBeta100Success}. The reconstructions are sufficiently accurate to allow accurate decryption, as shown in Figure~\ref{FIG:JointBeta100Success-Decryption}.
\begin{figure}[!htb]
    \centering
    \includegraphics[width=0.3\textwidth]{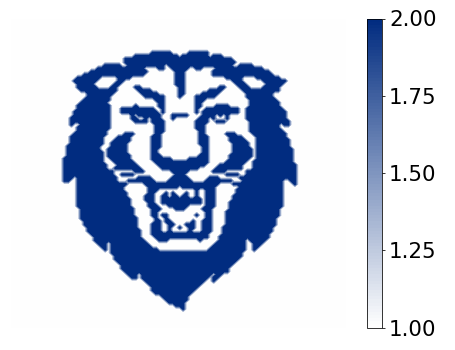}
    \includegraphics[width=0.3\textwidth]{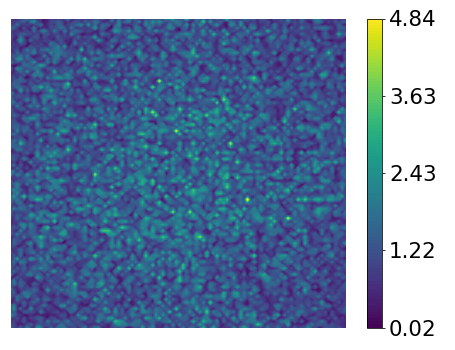}
    \includegraphics[width=0.3\textwidth]{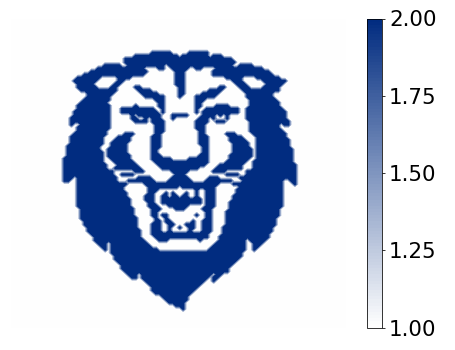}
    \caption{The true plaintext (left), amplitude of the ciphertext (middle), decrypted plaintext (right) using the reconstructed keys $(\wt\varphi_1, \wt\varphi_2, \wt\beta)$ in Numerical Experiment III.}
    \label{FIG:JointBeta100Success-Decryption}
\end{figure}

To provide some insight into the optimization process, we show the evolution of the relative $L^2$ errors of the reconstructions of $\varphi_1$ and $\beta$ in~\Cref{fig:convergence}. For $\beta$, the relative error is $\|\tilde{\beta}^{(i)} - \beta_{\mathrm{true}}\|_{L^2} / \|\beta_{\mathrm{true}}\|_{L^2}$, where $\tilde{\beta}^{(i)}$ denotes the reconstruction at iteration $i$. For $\varphi_1$, the error is computed up to an estimated constant phase shift (Proposition~\ref{prop:gauge}). 
The value of the data mismatch functional $\Phi$ decreases several orders of magnitude during the same evolution process.
\begin{figure}[!htb]
    \centering
    \begin{subfigure}[b]{0.48\textwidth}
        \centering
        \includegraphics[width=\textwidth]{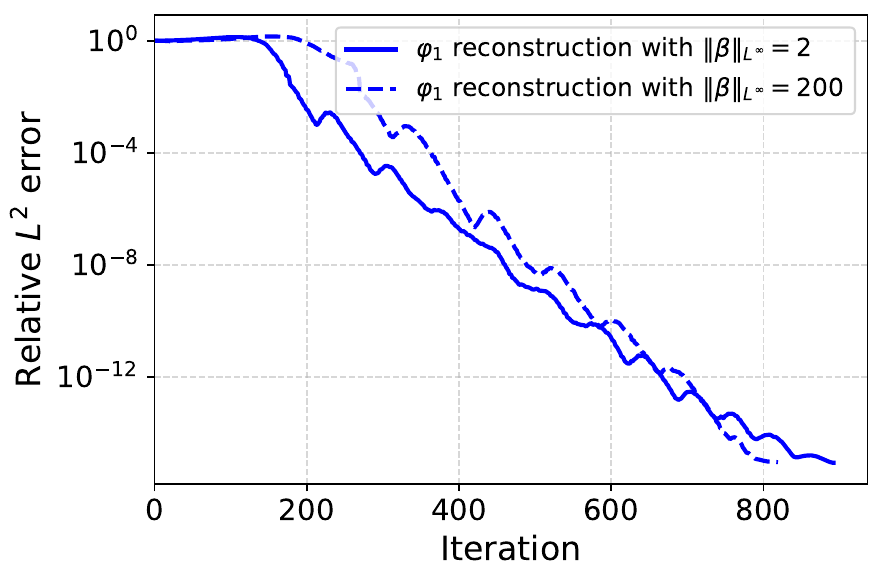}
        \label{fig:convergence_phi1}
    \end{subfigure}
    \hfill
    \begin{subfigure}[b]{0.48\textwidth}
        \centering
        \includegraphics[width=\textwidth]{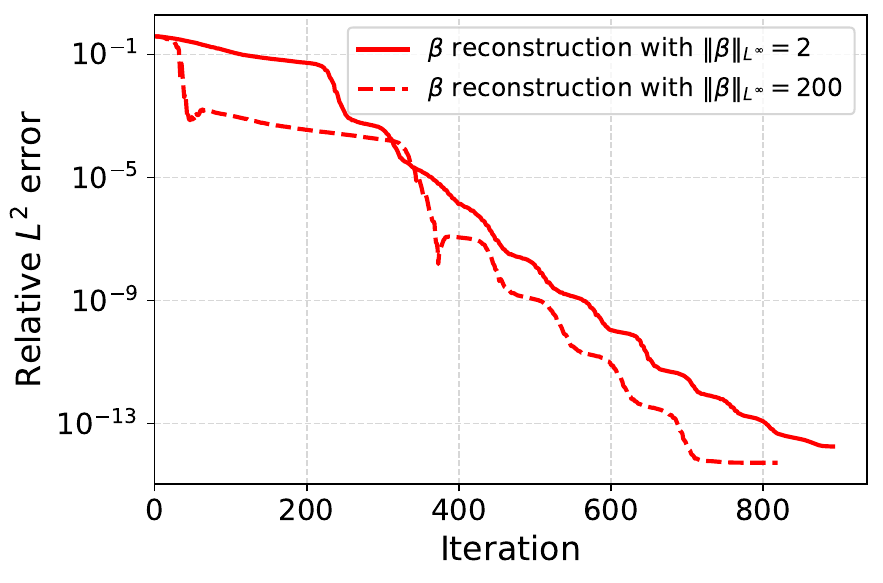}
        \label{fig:convergence_beta}
    \end{subfigure}
    \caption{Evolution of relative $L^2$ error of the reconstructions in the optimization process. Shown are the relative errors of $\varphi_1$ (left) and $\beta$ (right). Solid and dashed lines are for Numerical Experiment II and III, respectively.}
    \label{fig:convergence}
\end{figure}

\subsection{Effect of strong nonlinearity}
\label{sec:BetaMismatch}

It turns out that while the differential CPA strategy gives a uniqueness theory, as we have seen in~\Cref{sec:pointwise_attack} and~\Cref{SEC:Phase-Medium}, we observe in numerical simulations that the system gets harder to attack when the effective nonlinearity is very strong. The effective strength of the nonlinearity depends on the product of $\beta$ and the total propagation distance $L_z$, $|\beta| \cdot L_z$. When $|\beta|\cdot L_z$ gets large enough, the nonlinear optimization algorithm we developed has a hard time decreasing the objective functional $\Phi$. As examples, we show in~\Cref{FIG:JointBeta100Fail} typical reconstructions of the $(\varphi_1, \beta)$ pair in Numerical Experiment III with larger $L_z$ (and thus larger $|\beta|\cdot L_z$).
\begin{figure}[!htb]
    \centering
    \begin{subfigure}{\textwidth}
        \centering
        \begin{minipage}{0.24\textwidth}
            \includegraphics[width=\textwidth]{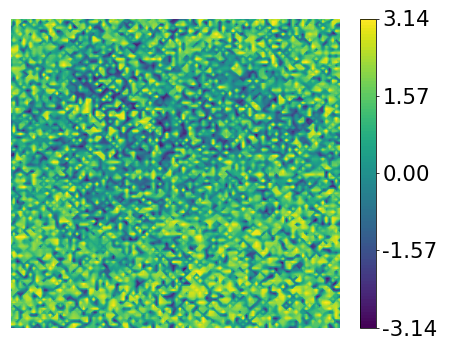}
            \subcaption*{$\tilde{\varphi}_1$ using $L_z=0.015$}
        \end{minipage}
        \hfill
        \begin{minipage}{0.24\textwidth}
            \includegraphics[width=\textwidth]{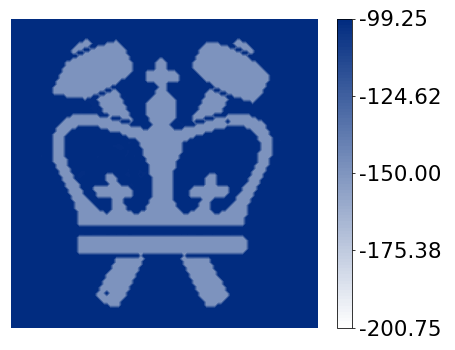}
            \subcaption*{$\tilde{\beta}$ using $L_z=0.015$}
        \end{minipage}
        \hfill
        \begin{minipage}{0.24\textwidth}
            \includegraphics[width=\textwidth]{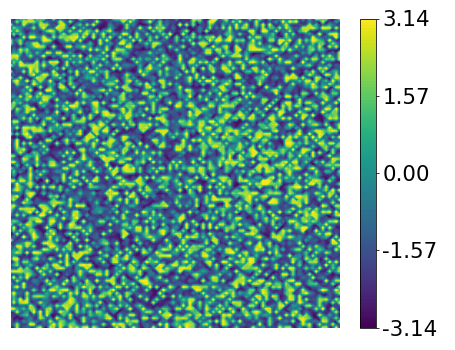}
            \subcaption*{$\tilde{\varphi}_1$ using $L_z=0.02$}
        \end{minipage}
        \hfill
        \begin{minipage}{0.24\textwidth}
            \includegraphics[width=\textwidth]{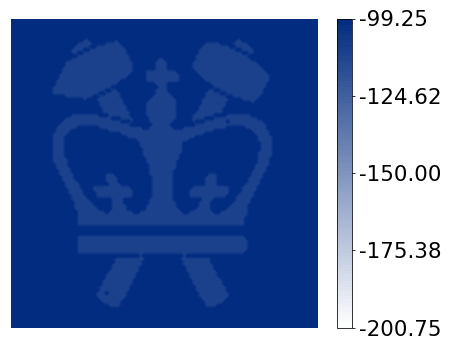}
            \subcaption*{$\tilde{\beta}$ using $L_z=0.02$}
        \end{minipage}
    \end{subfigure}
    \caption{Reconstruction of the $(\varphi_1, \beta)$ pair in Numerical Experiment III (in~\Cref{FIG:JointBeta100Success}) with larger propagation distance $L_z$ when optimization terminated prematurely. Shown are results for $L_z=0.015$ (left) and $L_z=0.02$ (right).}
    \label{FIG:JointBeta100Fail}
\end{figure}

To understand why the optimization stagnates at large $|\beta| \cdot L_z$, we examine one-dimensional cross sections of the objective function defined in \eqref{EQ:OBJ Func}. We fix the ground-truth parameters $(\varphi_1^*, \beta^*)$, where $\varphi_1^*$ is the random phase mask used throughout this section and $\beta^* = 1000 \cdot \beta_0$ ($\beta_0$ being the true $\beta$ shown in~\Cref{FIG:Exp3_Joint}).
We then draw a random unit vector $(\delta\varphi_1, \delta\beta)$ in the joint parameter space, and evaluate $\Phi((\varphi_1^*, \beta^*) + s \cdot (\delta\varphi_1,\delta\beta))$ as a function of the step size $s$. Figure~\ref{fig:cross_section} shows the resulting cross-sections for four values of $L_z$. At small $L_z$, the basin of attraction surrounding the global minimizer is broad, providing a large region in which gradient-based methods receive a useful descent signal. As $L_z$ increases, the basin narrows progressively while the surrounding plateau remains flat. This explains the stagnation observed in Figure~\ref{FIG:JointBeta100Fail}: the optimizer, initialized far from the true parameters, begins on the flat plateau where the gradient provides little directional information and terminates before reaching the narrow basin containing the global minimizer.
\begin{figure}[!htb]
    \centering
    \includegraphics[width=0.65\textwidth]{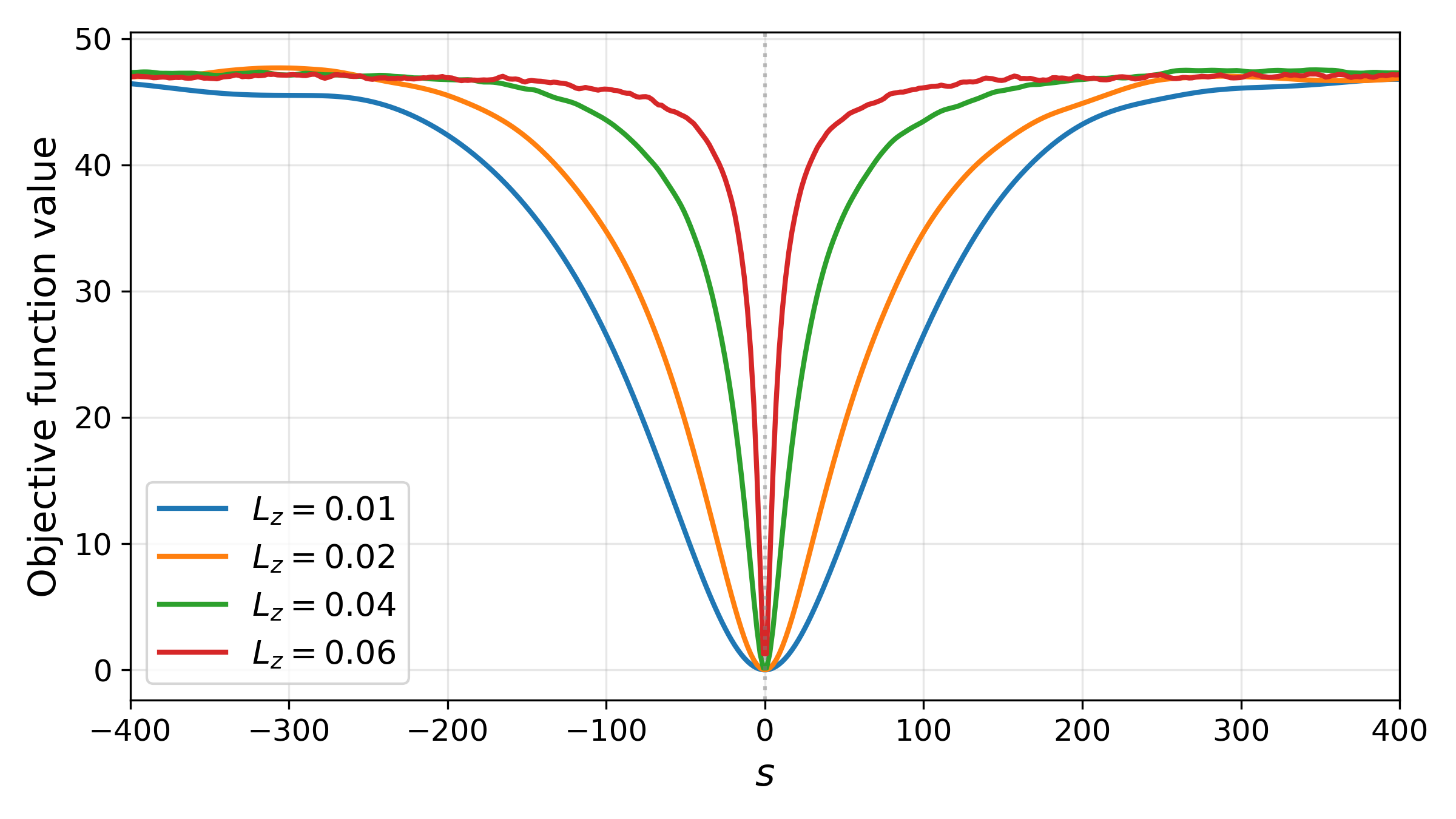}
    \caption{Value of $\Phi((\varphi_1^*, \beta^*) + s \cdot (\delta\varphi_1,\delta\beta))$ as a function of $s$ at a fixed random direction $(\delta\varphi_1,\delta \beta)$. Shown are plots for four different $L_z$ values. As $L_z$ increases, the basin of
    attraction narrows, requiring increasingly precise initialization for convergence.}
    \label{fig:cross_section}
\end{figure}

\paragraph{Numerical Experiment IV.} The main issue with the reconstructions in the case of large $|\beta|\cdot L_z$ is that the cross talk between $\beta$ and $\varphi_1$ becomes more significant. In other words, changing $\beta$ can largely compensate (even though not exactly reproduce) the effect of changing $\varphi_1$ vice versa.

In these experiments, the attacker uses $\beta_{\mathrm{model}} = 0$ throughout the reconstruction, while the true ciphertexts are generated with varying $\beta_{\mathrm{true}}$. To visualize potential cross talk between parameters, we use a structured phase mask $\varphi_1$ consisting of two circular anomalies.

We first consider $\beta \in [-2,-1]$ with propagation distance $L_z = 0.01$. Figure~\ref{FIG:BetaNeg1} shows that despite the model mismatch, the reconstruction succeeds: $\varphi_1$ is recovered accurately with minimal error (though some crosstalk can be seen in the error figure).

\begin{figure}[!htb]
    \centering
    \begin{subfigure}{\textwidth}
        \centering
        \begin{minipage}{0.24\textwidth}
            \includegraphics[width=\textwidth]{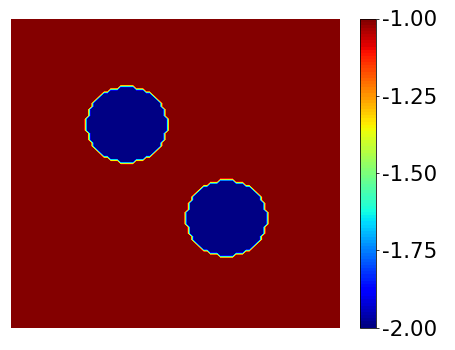}
            \subcaption*{True $\beta$}
        \end{minipage}
        \hfill
        \begin{minipage}{0.24\textwidth}
            \includegraphics[width=\textwidth]{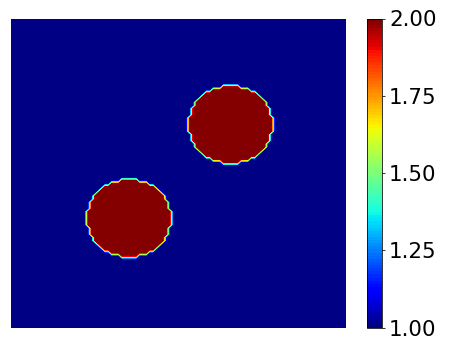}
            \subcaption*{True $\varphi_1$}
        \end{minipage}
        \hfill
        \begin{minipage}{0.24\textwidth}
            \includegraphics[width=\textwidth]{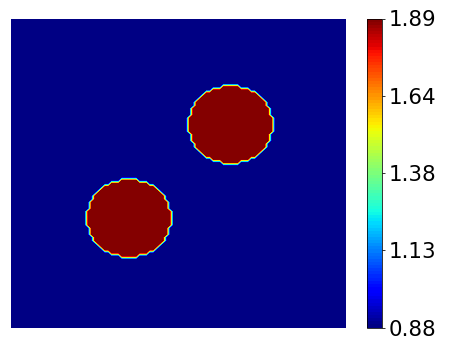}
            \subcaption*{Reconstructed $\tilde{\varphi}_1$}
        \end{minipage}
        \hfill
        \begin{minipage}{0.24\textwidth}
            \includegraphics[width=\textwidth]{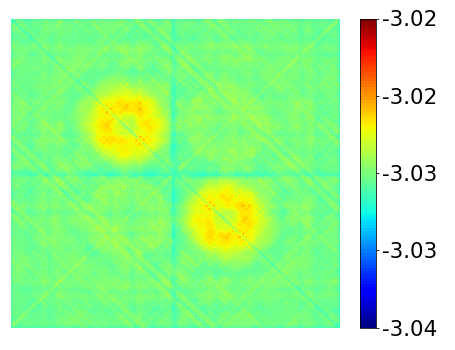}
            \subcaption*{Error in $\tilde{\varphi}_1$}
        \end{minipage}
    \end{subfigure}
    \caption{{\bf Model Mismatch, Case $\beta_{\mathrm{true}}=\beta$:} successful reconstruction despite ignoring the nonlinearity.}
    \label{FIG:BetaNeg1}
\end{figure}

However, at $\|\beta_{\mathrm{true}}\|_{L^\infty} = 200$ with the same propagation distance $L_z = 0.01$, reconstruction degrades significantly due to the model mismatch between the assumed linear dynamics ($\beta_{\mathrm{model}} = 0$) and the true nonlinear propagation. Figure~\ref{FIG:BetaModelMismatchFail} displays the reconstruction results for $\|\beta_{\mathrm{true}}\|_{L^\infty} = 200$ and $\|\beta_{\mathrm{true}}\|_{L^\infty} = 2000$. For the latter case, we have to reduce the propagation distance to $L_z = 0.001$ to ensure the optimization landscape remains well-conditioned enough for meaningful parameter updates, allowing the crosstalk phenomenon to manifest rather than having the optimization stagnate prematurely at a local minimum. For $\|\beta_{\mathrm{true}}\|_{L^\infty} = 200$, the optimization achieves only approximately one order of magnitude decrease in the objective function (from $1.65 \times 10^4$ to $2.10 \times 10^3$) before becoming trapped at a local minimum.

\begin{figure}[htb]
    \centering
    \begin{subfigure}{\textwidth}
        \centering
        \begin{minipage}{0.24\textwidth}
            \includegraphics[width=\textwidth]{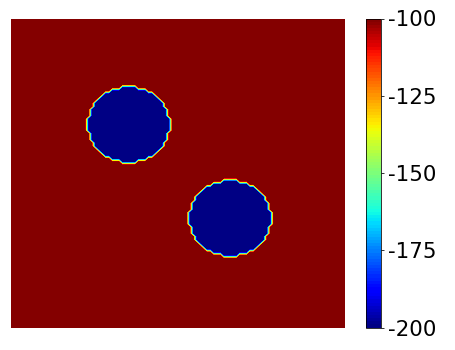}
            \subcaption*{True $\beta$}
        \end{minipage}
        \quad
        \begin{minipage}{0.24\textwidth}
            \includegraphics[width=\textwidth]{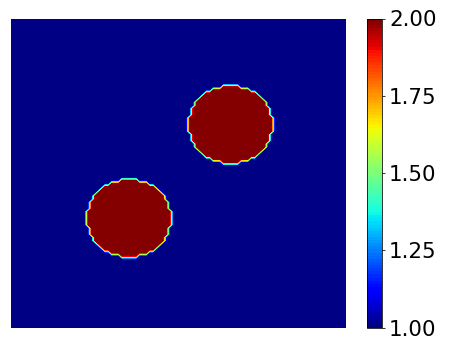}
            \subcaption*{True $\varphi_1$}
        \end{minipage}
        \quad
        \begin{minipage}{0.24\textwidth}
            \includegraphics[width=\textwidth]{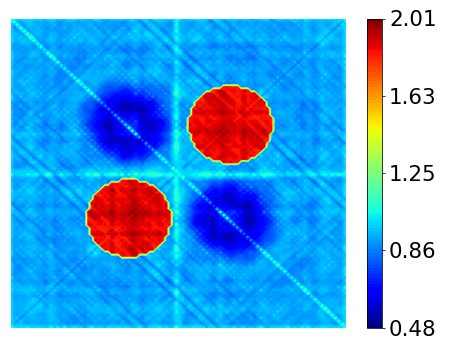}
            \subcaption*{Reconstructed $\tilde{\varphi}_1$}
        \end{minipage}
        \caption{$\beta_{\mathrm{true}}=100\cdot \beta$}
    \end{subfigure}
    
    \vspace{1em}
    
    \begin{subfigure}{\textwidth}
        \centering
        \begin{minipage}{0.24\textwidth}
            \includegraphics[width=\textwidth]{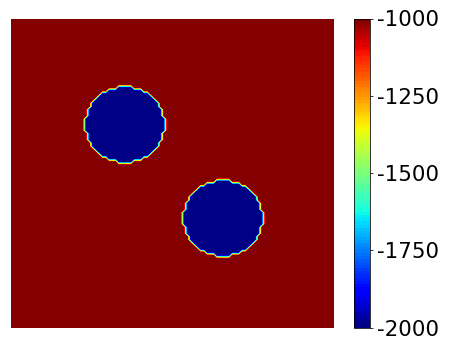}
            \subcaption*{True $\beta$}
       \end{minipage}
        \quad
        \begin{minipage}{0.24\textwidth}
            \includegraphics[width=\textwidth]{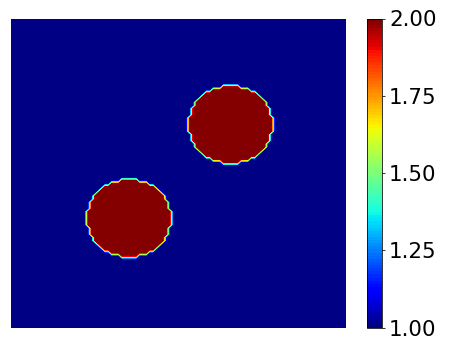}
            \subcaption*{True $\varphi_1$}
        \end{minipage}
        \quad
        \begin{minipage}{0.24\textwidth}
            \includegraphics[width=\textwidth]{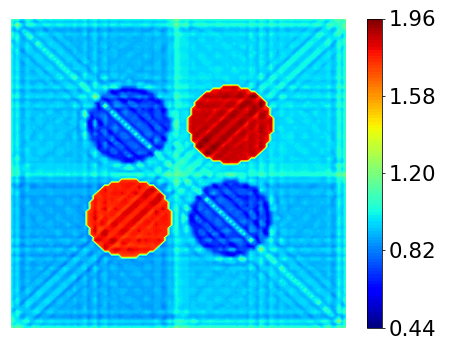}
            \subcaption*{Reconstructed $\tilde{\varphi}_1$}
        \end{minipage}
        \caption{$\beta_{\mathrm{true}}=1000\cdot \beta$}
    \end{subfigure}
    
    \caption{{\bf Model Mismatch:} Failed reconstruction when ignoring large nonlinearity. The recovered $\tilde{\varphi}_1$ reveal systematic patterns, indicating that $\varphi_1$ has compensated for the missing $\beta$ through cross talk.}
    \label{FIG:BetaModelMismatchFail}
\end{figure}

The recovered $\tilde{\varphi}_1$'s in Figure~\ref{FIG:BetaModelMismatchFail} reveal strong impact from the wrong $\beta$. This cross talk arises because the optimizer, constrained to use $\beta_{\mathrm{model}} = 0$, must compensate for the effects of the true nonlinearity through systematic adjustments to $\varphi_1$. The recovered phase mask is therefore biased, absorbing contributions that properly belong to $\beta$. This behavior is characteristic of inverse problems with model mismatch: when a relevant parameter is fixed at an incorrect value, the optimization distributes its effects among the remaining free parameters.

These results demonstrate one potential weakness in an encryption setting. As shown, at small $|\beta|$, the nonlinear contribution is negligible and may be safely ignored. At large $|\beta|$, however, an attacker using a linear model will obtain incorrect phase keys. This complements the findings of~\Cref{sec:JointBeta}: when $\|\beta\|_{L^\infty}$ is large, it must be accounted for in the reconstruction. \Cref{sec:JointBeta} showed that joint recovery of $(\varphi_1, \beta)$ is feasible when $\|\beta\|_{L^\infty} \cdot L_z$ is moderate; here we see that ignoring $\beta$ under the same conditions leads to errors in $\varphi_1$ due to model mismatch.

\section{Concluding remarks}

In this work, we studied the security of a nonlinear optical encryption system with double random phase under chosen plaintext attacks. We formulate the problem as an inverse problem to a nonlinear Schr\"{o}dinger equation with data encoded in the solution map of the PDE. We proposed a strategy of using differential data to linearize the forward propagation model so that the security keys, that is, the phase masks in the double random phase encryption, and the nonlinearity of the medium can be recovered from data at different orders of the multi-linearization process.

While on the mathematical level, adding the nonlinearity $\beta$ as an extra security key offers no additional benefit, our numerical simulation showed that strong nonlinearity can indeed increase the difficulty of CPA attacks. This means using more complicated nonlinearity can potentially improve the security of the encryption system against attacks.%

We also showed that the decryption process of the nonlinear encryption system is quite robust with respect to errors in the recovered security keys. In other words, small errors in the reconstruction of the security keys will result in small errors in the decrypted plaintext.

The work we have here is very preliminary. There are many important questions that need to be addressed. For instance, it would be great to have a more precise characterization of the impact of nonlinearity strength on the system's vulnerability to CPA attacks. It would also be interesting to study the uniqueness of the reconstruction of $(\varphi_1, \beta)$ from the data $f\mapsto |g|$ (that is, the uniqueness of the global minimizer of the data mismatch functional $\Phi$ in~\eqref{EQ:OBJ Func}). We also plan to analyze similar inverse problems for such nonlinear optical encryption systems under other attacks, such as known-plaintext attacks (KPA), where one has only a number of given plaintext-ciphertext pairs, and known-ciphertext attacks (KCA), where one has only a collection of ciphertexts, in future work.

\section*{Acknowledgments}

This work is partially supported by the National Science Foundation through grants DMS-1937254 and DMS-2309802. 
We would like to thank Columbia's High Performance Computing (HPC) service for the computational resources they provided for this project.

\appendix

\section{Details of computational implementation}
\label{APP:Implementation}

In this appendix, we provide details on the computational implementation of the CPA attack. The reconstruction of the security keys $\varphi_1$ and $\beta$ is based on the optimization problem
\begin{equation}\label{EQ:Min}
    (\wh \varphi_1, \wh \beta)=\argmin \Phi(\varphi_1,\beta)
\end{equation}
where the data mismatch functional is defined in~\eqref{EQ:OBJ Func}. Once $\varphi_1$ and $\beta$ have been recovered, we can use one of the ciphertexts, say $g_s(\bx)$, to recover the security key $\varphi_2$, by the formula
\[
    \wh \varphi_2=\arg\!\left(\frac{g_s(\bx)}{u_s(L_z,\bx)}\right).
\]
\subsection{Gradient computation}
\label{appendix:adjoint}

The main computational task for the optimization problem~\eqref{EQ:Min} is to compute the gradient of the data mismatch functional $\Phi$ with respect to its parameters, that is, the Fr\'echet derivative of $\Phi$ with respect to $\varphi_1$ and $\beta$. This is a standard process. We outline here the main process, which is similar to the calculation in~\cite{ChReSo-IP25}. 

For convenience, we work with the variable $\Theta=\exp(i\varphi_1)$ instead of $\varphi_1$. We optimize over $\Theta$ by separately varying its real and imaginary parts, $\Theta = \Theta_\Re + i \Theta_\Im$, which avoids the $2\pi$-ambiguity inherent in phase-based parameterizations.

We begin with the derivative with respect to $\beta$. We introduce the notations
    \[
        \mathcal{A}_s = \dfrac{|u_s|^2(2+|u_s|^2)}{(1+|u_s|^2)^2}, \quad 
        \mathcal{B}_s = \dfrac{u_s^2}{(1+|u_s|^2)^2}, \quad 
        r_s = 2\left(|u_s(L_z)|^2 - d_s^2\right)\overline{u_s(L_z)}.
    \]
By the chain rule, we have
    \begin{multline*}
        \Phi'(\beta) [\delta \beta] = \sum_{s=1}^{N_s}\int_\Omega 2\left(|u_s(L_z,\bx)|^2 - d_s(\bx)^2\right) \Re\left[\overline{u_s(L_z, \bx)}\,\delta u_s(L_z, \bx)\right] d\bx \\ = \sum_{s=1}^{N_s}\int_\Omega  \Re\left[r_s \,\delta u_s(L_z, \bx) \right] d\bx.
    \end{multline*}
    Meanwhile, we linearize the NLSE equation with respect to $\beta$, denoting by $\delta u_s:=u_s'[\delta\beta]$ the Fr\'echet derivative of $u_s$ at $\beta$ in the $\delta\beta$ direction, to get,
    \[
        i\frac{\partial (\delta u_s)}{\partial z} + \frac{1}{2k} \Delta (\delta u_s) + \beta \left( \mathcal{A}_s\, \delta u_s + \mathcal{B}_s\, \overline{\delta u_s} \right) = - \frac{|u_s|^2}{1+|u_s|^2}\, u_s\,\delta \beta,
    \]
    with $\delta u_s(0,\bx) = 0$ and periodic boundary conditions. Let us define the linearized operators $\cL$ and $\cL^*$
    \begin{align*}
        \cL(f) &= i\frac{\partial f}{\partial z} + \frac{1}{2k} \Delta f + \beta \left( \mathcal{A}_s\, f + \mathcal{B}_s\, \bar{f} \right), \\
        \cL^*(g) &= -i\frac{\partial g}{\partial z} + \frac{1}{2k} \Delta g + \beta \left( \mathcal{A}_s\, g + \mathcal{B}_s\, \bar{g} \right).
    \end{align*}
    with periodic boundary conditions. Then it is straightforward to verify that
    \[
        \int_{0}^{L_z} \int_\Omega \left( f\,\cL^* g - g\,\cL f \right) d\bx\, dz = -i \int_\Omega [f\,g]_{z=0}^{L_z}\, d\bx.
    \]
    Let $v_s$ be the solution to the adjoint equation $\cL^*(v_s) = 0$ with terminal condition $v_s(L_z,\bx) = r_s(\bx)$, that is, $v_s':=v_s(L_z-z,\bx)$ solves the problem:
    \begin{equation}\label{EQ:Adjoint}
        \begin{array}{rcll}
        i\dfrac{\partial v_s'}{\partial z} + \dfrac{1}{2k}\Delta v_s' + \beta \left( \mathcal{A}_s v_s' + \mathcal{B}_s \overline{v}_s \right) &=& 0, & \mbox{in}\ \ (0, L_z]\times \Omega,\\[2ex]
        v_s'(0, \bx)&=& r_s(\bx), & \mbox{in}\ \ \Omega,
        \end{array}
    \end{equation}
        with periodic boundary conditions. 
    Then we verify that
    \begin{align*}
        \int_\Omega r_s\, \delta u_s(L_z,\bx) \, d\bx 
        &= \int_\Omega v_s'(L_z,\bx)\, \delta u_s(L_z,\bx) \, d\bx - \int_\Omega v_s'(0,\bx)\, \delta u_s(0,\bx) \, d\bx \\
        &= i \int_{0}^{L_z} \int_\Omega v_s'\,\cL(\delta u_s) \, d\bx\, dz 
        = i \int_{0}^{L_z} \int_\Omega v_s'\, \frac{|u_s|^2}{1+|u_s|^2}\, u_s\, \delta \beta \, d\bx\, dz.
    \end{align*}
    Taking real parts of this relation leads us to
    \[
        \int_\Omega \Re \Big(r_s\, \delta u_s(L_z,\bx) \Big)\, d\bx 
        = \int_\Omega \int_{0}^{L_z} \Re\Big(i v_s'\, \frac{|u_s|^2}{1+|u_s|^2}\, u_s\Big) dz\, \delta \beta \, d\bx\, dz\,.
    \]

    The above calculation has, therefore, led to, after using $v_s'(z,\bx)=v_s(L_z-z,\bx)$,
    \begin{equation}\label{EQ:Deri Phi beta}
    \Phi'(\beta) [\delta \beta] =\sum_{s=1}^{N_s} \int_\Omega \left(\int_{0}^{L_z} \Re\left[i\,v_s(L_z - z,\bx) \,\frac{|u_s|^2 u_s}{1+|u_s|^2}\right] dz \right)\delta \beta\,d\bx
    \end{equation}
    
The derivative with respect to $\Theta$ is largely similar. Varying the initial condition $u_s(0,\bx) = f_s\, \Theta$ with $\delta u_s(0,\bx) = f_s\,\delta \Theta$ yields the homogeneous linearized equation $\cL(\delta u_s) = 0$. The adjoint identity reduces to
    \[
        \int_\Omega v_s'(L_z,\bx)\, \delta u_s(L_z,\bx) \, d\bx = \int_\Omega v_s'(0,\bx)\, \delta u_s(0,\bx) \, d\bx.
    \]
    With $v_s'(L_z,\bx) = r_s$ and $\delta u_s(0,\bx) = f_s \,\delta \Theta$, we obtain
    \[
        \Phi'(\Theta)[\delta \Theta] = \sum_{s=1}^{N_s}\Re \int_\Omega r_s\, \delta u_s(L_z)\, d\bx = \sum_{s=1}^{N_s} \Re \int_\Omega v_s'(0,\bx)\, f_s\, \delta \Theta\, d\bx.
    \]
    Writing $\delta \Theta = \delta \Theta_\Re + i\,\delta \Theta_\Im$ and $v_s'(0,\bx) = v_{s,\Re}' + i\, v_{s,\Im}'$, we have
    \[
        \Re\left[ v_s(0,\bx)\, f_s\, \delta \Theta \right] = f_s\, v_{s,\Re}'\, \delta \Theta_\Re - f_s\, v_{s,\Im}'\, \delta \Theta_\Im.
    \]
    In the reversed coordinates, $v_s(L_z-z,\bx) = v_s'(z,\bx)$, which gives
    \begin{align}
        \label{EQ:Deri Phi Rr}
        \Phi'(\Theta_\Re) [\delta \Theta_\Re]  &= \sum_{s=1}^{N_s} \int_\Omega f_s(\bx)\, \Re\left[v_s(L_z,\bx)\right] \delta \Theta_\Re\,d\bx, \\
        \label{EQ:Deri Phi Ri}
        \Phi'(\Theta_\Im) [\delta \Theta_\Im]  &= -\sum_{s=1}^{N_s} \int_\Omega f_s(\bx)\, \Im\left[v_s(L_z,\bx)\right] \delta \Theta_\Im\,d\bx,
    \end{align} 
The gradient of the data mismatch functional can therefore be computed using formulas~\eqref{EQ:Deri Phi beta}, ~\eqref{EQ:Deri Phi Rr}, and~\eqref{EQ:Deri Phi Ri}. Therefore, each evaluation of the gradient costs one solution of the forward problem, and one solution of the adjoint problem, i.e.,~\eqref{EQ:Adjoint}.

\subsection{Optimization and discretization}

In our implementation, we solve the optimization problem~\eqref{EQ:Min} using a quasi-Newton method with the limited memory version of BFGS updating rule for the approximated Hessian. This algorithm is coded in Python's \texttt{scipy.optimize} function \cite{2020SciPy-NMeth}. The optimizer uses
\[
    \texttt{ftol} = 10^{-50}, \qquad
    \texttt{gtol} = 10^{-15}, \qquad
    \texttt{maxls} = 100,
\]
and a prescribed maximum iteration count \texttt{MaxIT}.  
A strong Wolfe line search is employed. We only need to provide the gradient information needed by the algorithm.

In all the numerical simulations, we take the computational domain to be
\[
    \Omega = (0,L_x)\times(0,L_y), \qquad L_x = L_y = 1 .
\]
A uniform Cartesian mesh is employed with
\[
    N_x = N_y = 100, \qquad 
    \Delta x = \Delta y = 0.01 .
\]
Periodic boundary conditions are imposed in both spatial directions.  A standard finite difference scheme is used to discretize the forward and adjoint equations. After spatial discretization, the nonlinear Schr\"odinger equation takes the semi-discrete form
\[
    \frac{d u}{dz} = F(u),
\]
where
\[
    F(u) = i\frac{1}{2k}\Delta u -i \beta\,\frac{|u|^2}{1+|u|^2}\,u.
\]
The update in $z$ is computed by the classical fourth–order Runge–Kutta scheme. The adjoint equation~\eqref{EQ:Adjoint} is discretized in $z$ in the same fashion.

{\small
\bibliography{BIB-REN,BIB-Encryption}

\begin{thebibliography}{10}

\bibitem{Abuturab-OLE24}
{\sc M.~R. Abuturab}, {\em Multilevel information cryptosystem using
  generalized singular value decomposition, optical interference, and devil's
  vortex {Fresnel} lens encoding}, Optics and Lasers in Engineering, 181
  (2024), p.~108399.

\bibitem{AhCa-OLE23}
{\sc K.~Ahmadi and A.~Carnicer}, {\em Optical visual encryption using focused
  beams and convolutional neural networks}, Optics and Lasers in Engineering,
  161 (2023), p.~107321.

\bibitem{AlJaHuMu-OQE23}
{\sc A.~Ali, S.~Javed, R.~Hussain, and T.~Muhammad}, {\em Secure information
  transmission using the fractional coupled {Schr{\"o}dinger} model: a
  dynamical perspective}, Optical and Quantum Electronics, 55 (2023), p.~1267.

\bibitem{CaMoArJu-OL05}
{\sc A.~Carnicer, M.~Montes-Usategui, S.~Arcos, and I.~Juvells}, {\em
  Vulnerability to chosen-cyphertext attacks of optical encryption schemes
  based on double random phase keys}, Optics Letters, 30 (2005),
  pp.~1644--1646.

\bibitem{ChReSo-IP25}
{\sc Y.~Cheng, K.~Ren, and N.~Soedjak}, {\em Phase retrieval via media
  diversity}, Inverse Problems, 41 (2025).
\newblock 075009.

\bibitem{deMoPe-CVPDE13}
{\sc L.~de~Almeida~Maia, E.~Montefusco, and B.~Pellacci}, {\em Weakly coupled
  nonlinear schr{\"o}dinger systems: the saturation effect}, Calculus of
  Variations and Partial Differential Equations, 46 (2013), pp.~325--351.

\bibitem{GaHe-JOSA91}
{\sc S.~Gatz and J.~Herrmann}, {\em Soliton propagation in materials with
  saturable nonlinearity}, Journal of the Optical Society of America B, 8
  (1991), pp.~2296--2302.

\bibitem{HoSi-EL22}
{\sc J.~Hou and G.~Situ}, {\em Image encryption using spatial nonlinear
  optics}, elight, 2 (2022), p.~3.

\bibitem{HuChGoZh-OLE20}
{\sc Z.-J. Huang, S.~Cheng, L.-H. Gong, and N.-R. Zhou}, {\em Nonlinear optical
  multi-image encryption scheme with two-dimensional linear canonical
  transform}, Optics and Lasers in Engineering, 124 (2020), p.~105821.

\bibitem{InCh-JVCIR21}
{\sc K.~Inoue and M.~Cho}, {\em Amplitude based keyless optical encryption
  system using deep neural network}, Journal of Visual Communication and Image
  Representation, 79 (2021), p.~103251.

\bibitem{Isakov-ARMA93}
{\sc V.~Isakov}, {\em On uniqueness in inverse problems for semilinear
  parabolic equations}, Arch. Rational Mech. Anal., 124 (1993), pp.~1--12.

\bibitem{KaSiKa-MPE21}
{\sc M.~Kaur, S.~Singh, and M.~Kaur}, {\em Computational image encryption
  techniques: a comprehensive review}, Mathematical Problems in Engineering,
  2021 (2021), p.~5012496.

\bibitem{LaLuZh-arXiv23}
{\sc R.-Y. Lai, X.~Lu, and T.~Zhou}, {\em Partial data inverse problems for the
  nonlinear {Schr{\"o}dinger} equation}, arXiv:2306.15935,  (2023).

\bibitem{LaUhYa-arXiv24}
{\sc R.-Y. Lai, G.~Uhlmann, and L.~Yan}, {\em Partial data inverse problems for
  the nonlinear magnetic schr\"{o}dinger equation}, arXiv:2411.06369,  (2024).

\bibitem{LaOkSaSaTe-arXiv24}
{\sc M.~Lassas, L.~Oksanen, S.~K. Sahoo, M.~Salo, and A.~Tetlow}, {\em
  Coefficient determination for non-linear {Schr{\"o}dinger} equations on
  manifolds}, arXiv:2201.03699,  (2024).

\bibitem{LiLi-CTP21}
{\sc J.~Li and B.~Li}, {\em Solving forward and inverse problems of the
  nonlinear {Schr{\"o}dinger} equation with the generalized-symmetric scarf-ii
  potential via {PINN} deep learning}, Communications in Theoretical Physics,
  73 (2021), p.~125001.

\bibitem{LiGuSh-OLT}
{\sc S.~Liu, C.~Guo, and J.~T. Sheridan}, {\em A review of optical image
  encryption techniques}, Optics \& Laser Technology, 57 (2014), pp.~327--342.

\bibitem{MoChKaJaCh-NP23}
{\sc J.~Moon, Y.-C. Cho, S.~Kang, M.~Jang, and W.~Choi}, {\em Measuring the
  scattering tensor of a disordered nonlinear medium}, Nature Physics, 19
  (2023), pp.~1709--1718.

\bibitem{PeWeZh-OL06}
{\sc X.~Peng, H.~Wei, and P.~Zhang}, {\em Chosen-plaintext attack on lensless
  double-random phase encoding in the {Fresnel} domain}, Optics Letters, 31
  (2006), pp.~3261--3263.

\bibitem{PeZhWeYu-OL06}
{\sc X.~Peng, P.~Zhang, H.~Wei, and B.~Yu}, {\em Known-plaintext attack on
  optical encryption based on double random phase keys}, Optics Letters, 31
  (2006), pp.~1044--1046.

\bibitem{ReJa-OL95}
{\sc P.~Refregier and B.~Javidi}, {\em Optical image encryption based on input
  plane and {Fourier} plane random encoding}, Optics Letters, 20 (1995),
  pp.~767--769.

\bibitem{ReSoWa-JDE25}
{\sc K.~Ren, N.~Soedjak, and K.~Wang}, {\em Unique determination of cost
  functions in a multipopulation mean field game model}, J. Diff. Eqn., 427
  (2025), pp.~843--867.

\bibitem{SaDiGoFi-PhysicaD20}
{\sc A.~Sagiv, A.~Ditkowski, R.~H. Goodman, and G.~Fibich}, {\em Loss of
  physical reversibility in reversible systems}, Physica D: Nonlinear
  Phenomena, 410 (2020), p.~132515.

\bibitem{SiZh-OL04}
{\sc G.~Situ and J.~Zhang}, {\em Double random-phase encoding in the {Fresnel}
  domain}, Optics Letters, 29 (2004), pp.~1584--1586.

\bibitem{ToBa-IEEE03}
{\sc M.~Tomas-Rodriguez and S.~Banks}, {\em The dynamical inverse problem for a
  nonlinear {Schrodinger} equation using boundary control}, in 2003 IEEE
  International Workshop on Workload Characterization, vol.~3, 2003,
  pp.~732--735.

\bibitem{Ton-AAA02}
{\sc B.~A. Ton}, {\em An inverse problem for a nonlinear {Schr{\"o}dinger}
  equation}, Abstract and Applied Analysis, 7 (2002), pp.~385--399.

\bibitem{TuPrLeWaFrKaDe-Optica17}
{\sc S.~K. Turitsyn, J.~E. Prilepsky, S.~T. Le, S.~Wahls, L.~L. Frumin,
  M.~Kamalian, and S.~A. Derevyanko}, {\em Nonlinear {Fourier} transform for
  optical data processing and transmission: advances and perspectives}, Optica,
  4 (2017), pp.~307--322.

\bibitem{UnJoSi-IL00}
{\sc G.~Unnikrishnan, J.~Joseph, and K.~Singh}, {\em Optical encryption by
  double-random phase encoding in the fractional fourier domain}, Optics
  letters, 25 (2000), pp.~887--889.

\bibitem{2020SciPy-NMeth}
{\sc P.~Virtanen, R.~Gommers, T.~E. Oliphant, M.~Haberland, T.~Reddy,
  D.~Cournapeau, E.~Burovski, P.~Peterson, W.~Weckesser, J.~Bright, S.~J. {van
  der Walt}, M.~Brett, J.~Wilson, K.~J. Millman, N.~Mayorov, A.~R.~J. Nelson,
  E.~Jones, R.~Kern, E.~Larson, C.~J. Carey, {\.I}.~Polat, Y.~Feng, E.~W.
  Moore, J.~{VanderPlas}, D.~Laxalde, J.~Perktold, R.~Cimrman, I.~Henriksen,
  E.~A. Quintero, C.~R. Harris, A.~M. Archibald, A.~H. Ribeiro, F.~Pedregosa,
  P.~{van Mulbregt}, and {SciPy 1.0 Contributors}}, {\em {{SciPy} 1.0:
  Fundamental Algorithms for Scientific Computing in Python}}, Nature Methods,
  17 (2020), pp.~261--272.

\bibitem{WaNiWaZhHu-OLE22}
{\sc F.~Wang, R.~Ni, J.~Wang, Z.~Zhu, and Y.~Hu}, {\em Invertible encryption
  network for optical image cryptosystem}, Optics and Lasers in Engineering,
  149 (2022), p.~106784.

\bibitem{WaLi-arXiv21}
{\sc Y.~Wang and Z.~Li}, {\em Inverse problem of nonlinear {Schr\"{o}dinger}
  equation as learning of convolutional neural network}, arXiv:2107.08593,
  (2021).

\bibitem{YaMu-JMS99}
{\sc G.~Y. Yagubov and M.~Musaeva}, {\em The inverse problem for a nonlinear
  {Schr{\"o}dinger} equation}, J. Math. Sci., 97 (1999), pp.~3981--3984.

\bibitem{YuMaXiSuLiJiWa-JMO22}
{\sc X.~Yu, M.~Ma, J.~Xiao, Y.~Sun, X.~Li, P.~Jiang, and K.~Wang}, {\em
  Scattering-medium-based optical image encryption by chaos and digital optical
  phase conjugation}, Journal of Modern Optics, 69 (2022), pp.~1006--1015.

\bibitem{Yu-NC24}
{\sc Z.~Yu, H.~Li, W.~Zhao, P.-S. Huang, Y.-T. Lin, J.~Yao, W.~Li, Q.~Zhao,
  P.~C. Wu, B.~Li, P.~Genevet, Q.~Song, and P.~Lai}, {\em High-security
  learning-based optical encryption assisted by disordered metasurface}, Nature
  Communications, 15 (2024), p.~2607.

\bibitem{ZeBeMiLuCaLiZhZe-OE25}
{\sc X.~Zeng, M.~R. Beli{\'c}, D.~Mihalache, X.~Lu, Y.~Cai, J.~Li, X.~Zhu, and
  L.~Zeng}, {\em Two-dimensional soliton families in saturable quasi-nonlinear
  lattices}, Optics Express, 33 (2025), pp.~33483--33493.

\bibitem{ZhLiZhXuXuWa-OE23}
{\sc L.~Zhang, S.~Lin, Q.~Zhou, J.~Xue, B.~Xu, and X.~Wang}, {\em Speckle-based
  optical encryption with complex-amplitude coding and deep learning}, Optics
  Express, 31 (2023), pp.~35293--35304.

\bibitem{ZhTaLi-OE21}
{\sc P.~Zheng, Q.~Tan, and H.-c. Liu}, {\em Inverse computational ghost imaging
  for image encryption}, Optics Express, 29 (2021), pp.~21290--21299.

\bibitem{ZhXiCh-OL20}
{\sc L.~Zhou, Y.~Xiao, and W.~Chen}, {\em Learning complex scattering media for
  optical encryption}, Optics Letters, 45 (2020), pp.~5279--5282.

\bibitem{ZhXiCh-OLE21}
\leavevmode\vrule height 2pt depth -1.6pt width 23pt, {\em Learning-based
  optical authentication in complex scattering media}, Optics and Lasers in
  Engineering, 141 (2021), p.~106570.

\bibitem{ZhZhWaXuXuZh-OLE23}
{\sc Q.~Zhou, L.~Zhang, X.~Wang, B.~Xu, J.~Xue, and Y.~Zhang}, {\em Optical
  encryption using a sparse-data-driven framework}, Optics and Lasers in
  Engineering, 171 (2023), p.~107825.

\bibitem{ZhYaLiLvZhChKeQiSh-OC23}
{\sc Y.~Zhu, D.~Yang, Z.~Li, W.~Lv, J.~Zhang, H.~Chen, C.~Ke, J.~Qiu, and
  Y.~Shi}, {\em Modified optical multi-image hiding method with a
  physics-driven neural network}, Optics Communications, 537 (2023), p.~129367.

\end{thebibliography}
\bibliographystyle{siam}
}

\end{document}